\documentclass{amsart}

\usepackage{amsmath,amsfonts,amsthm,amssymb,graphicx,float,tikz,indentfirst,graphics,setspace,young,mathrsfs,bbm}
\usepackage{verbatim, url, multirow,fullpage,mathtools}
\usepackage[hidelinks]{hyperref}
\usepackage[noend]{algpseudocode}
\usepackage[enableskew,vcentermath]{youngtab}
\usetikzlibrary{decorations.pathreplacing,patterns}

\usepackage{longtable, rotating,makecell,array}
\usepackage[aligntableaux=top]{ytableau}
\usepackage{soul}
\usepackage[colorinlistoftodos,textsize=footnotesize]{todonotes}

\setstcolor{red}

\usepackage[all]{xy}
\xyoption{poly}
\xyoption{arc}

\usetikzlibrary{decorations.pathmorphing,calc}
\usetikzlibrary{intersections}

\definecolor{light-gray}{gray}{0.6}

\tikzstyle{propagator}=[decorate,decoration={snake,amplitude=0.8mm}]
\tikzstyle{smallpropagator}=[decorate,decoration={snake,segment length=3mm,amplitude=0.5mm}]

\tikzstyle{firstdash}=[dashed,line cap=round, dash pattern=on 2pt off 1pt]
\tikzstyle{seconddash}=[dashed,line cap=round, dash pattern=on 0.5pt off 1pt]

\newcommand{\drawWLD}[2]{

\pgfmathsetmacro{\n}{#1}
\pgfmathsetmacro{\radius}{#2}
\pgfmathsetmacro{\angle}{360/\n}
    \foreach \i in {1,2,...,\n} {
      \pgfmathsetmacro{\x}{\angle*\i}
        \draw[-,shorten >=-\radius*0.1 cm,shorten <=-\radius*0.1 cm]  (\x:\radius cm)-- (\x + \angle: \radius cm);
    }

\draw (\angle:\radius) node {$\bullet$};
}

\newcommand{\drawprop}[4]{
\pgfmathsetmacro{\r}{#1}
\pgfmathsetmacro{\bumpr}{#2}
\pgfmathsetmacro{\s}{#3}
\pgfmathsetmacro{\bumps}{#4}
\pgfmathsetmacro{\perturbe}{\angle/\n}

\begin{scope}
\clip (\angle*\r:\radius) -- (\angle + \angle*\r:\radius) -- (\angle*\s:\radius) -- (\angle + \angle*\s:\radius) -- (\angle*\r:\radius);
\draw[propagator] (\angle*\r + \angle/2 + \bumpr*\perturbe:\radius) -- (\angle*\s + \angle/2 + \bumps*\perturbe:\radius);
\end{scope}
}

\newcommand{\modifiedprop}[5]{
\pgfmathsetmacro{\r}{#1}
\pgfmathsetmacro{\bumpr}{#2}
\pgfmathsetmacro{\s}{#3}
\pgfmathsetmacro{\bumps}{#4}
\pgfmathsetmacro{\perturbe}{\angle/\n}

\begin{scope}
\clip (\angle*\r:\radius) -- (\angle + \angle*\r:\radius) -- (\angle*\s:\radius) -- (\angle + \angle*\s:\radius) -- (\angle*\r:\radius);
\draw[#5] (\angle*\r + \angle/2 + \bumpr*\perturbe:\radius) -- (\angle*\s + \angle/2 + \bumps*\perturbe:\radius);
\end{scope}
}

\newcommand{\boundaryprop}[4]{
\pgfmathsetmacro{\r}{#1}
\pgfmathsetmacro{\bumpr}{#2}
\pgfmathsetmacro{\s}{#3}
\pgfmathsetmacro{\perturbe}{\angle/\n}

\begin{scope}
\clip (\angle*\r:\radius) -- (\angle + \angle*\r:\radius) -- (\angle*\s - \angle:\radius) -- (\angle*\s:\radius) -- (\angle + \angle*\s:\radius) -- (\angle*\r:\radius);
\draw[#4] (\angle*\r + \angle/2 + \bumpr*\perturbe:\radius) -- (\angle*\s:\radius);
\end{scope}
	
}

\newcommand{\drawnumbers}{
  \foreach \i in {1,2,...,\n} {
  \pgfmathsetmacro{\x}{\angle*\i}
  \draw (\x:\radius*1.15) node {\footnotesize \i};
}
}

\newcommand{\boundA}[3]{
	\pgfmathsetmacro{\r}{#1}
	\pgfmathsetmacro{\bumpr}{#2}
	\pgfmathsetmacro{\destination}{#3}
	\pgfmathsetmacro{\perturbe}{\angle/\n}
	\path [name path=polyedge1] (\angle*\r:\radius) -- (\angle*\r + \angle:\radius);
	\path [name path=radius1] (0:0) -- (\angle*\r + \angle/2 + \bumpr*\perturbe:\radius);
	\draw[->,
	name intersections={of=polyedge1 and radius1,by=p},
	shorten >=\radius*0.1 cm] (p) ++(\angle*\r + \angle/2 + \bumpr*\perturbe:\radius*0.15) -- (\angle*\destination: \radius*1.15);

}

\newcommand{\boundB}[3]{
	\pgfmathsetmacro{\rangle}{#1*\angle + \angle/2 + #2*\angle/\n}
	\pgfmathsetmacro{\sangle}{#1*\angle + \angle/2 + #3*\angle/\n}

	\draw[->,shorten <=\radius*0.02cm,shorten >=\radius*0.05cm] (\rangle:\radius*1.05) -- (\sangle:\radius*1.05);

}

\newcommand{\makediag}[8]{
	\begin{tikzpicture}[rotate=60,baseline=(current bounding box.east)]
	\begin{scope}
	\drawWLD{6}{0.8}
	\drawprop{#1}{#2}{#3}{#4}
	\drawprop{#5}{#6}{#7}{#8}
	\end{scope}
	\end{tikzpicture}
}

\newcommand{\R}{\mathbb{R}}

\newcommand{\Grnn}{\textrm{Gr}_{\R, \geq 0}}
\newcommand{\Gr}{\textrm{Gr}_{\R}}

\newcommand{\rr}{\mathbb{R}}

\newcommand{\D}{\partial}
\newcommand{\rk}{\textrm{rk }}

\def\ba #1\ea{\begin{align} #1 \end{align}}
\def\bas #1\eas{\begin{align*} #1 \end{align*}}
\def\bml #1\eml{\begin{multline} #1 \end{multline}}
\def\bmls #1\emls{\begin{multline*} #1 \end{multline*}}

\newcommand{\cP}{\mathcal{P}}

\newcommand{\fS}{\mathfrak{S}}

\newcommand{\cA}{\mathcal{A}}
\newcommand{\cI}{\mathcal{I}}

\newcommand{\cD}{\mathcal{D}}

\newcommand{\cE}{\mathcal{E}}
\newcommand{\cV}{\mathcal{V}}

\newcommand{\Prop}{\textrm{Prop}}

\newcommand{\cW}{\mathcal{W}}

\newcommand{\cZ}{\mathcal{Z}}

\newcommand{\gale}[1]{\preccurlyeq_{#1}}
\newcommand{\sgale}[1]{\prec_{#1}}
\newcommand{\Le}{\reflectbox{L}}
\newcommand{\proj}{\textrm{proj}}

\newtheorem{thm}{Theorem}[section]
\newtheorem{conj}[thm]{Conjecture}
\newtheorem{lem}[thm]{Lemma}
\newtheorem{cor}[thm]{Corollary}
\newtheorem{prop}[thm]{Proposition}
\newtheorem{algorithm}[thm]{Algorithm}

\theoremstyle{remark}
\newtheorem{eg}[thm]{Example}

\theoremstyle{dfn}
\newtheorem{dfn}[thm]{Definition}
\newtheorem{rmk}[thm]{Remark}

\begin{document}

\title{Wilson loops in SYM $N=4$ do not parametrize an orientable space}
\author{Susama Agarwala}
\address[Agarwala]{Mathematics Department \\ Chauvenet Hall\\ 572C Holloway Road\\ Annapolis, MD 21402-5002}

\email[Agarwala]{susama@alum.mit.edu}

\author{Cameron Marcott}
\address[Marcott]{Department of Mathematics\\ University of Waterloo \\
Waterloo, Ontario \\ Canada N2L 3G1}
\email[Marcott]{c2marcott@uwaterloo.ca}

\date{\today}
\maketitle

\begin{abstract}
In this paper we explore the geometric space parametrized by (tree level) Wilson loops in SYM $N=4$. We show that, this space can be seen as a vector bundle over a totally non-negative subspace of the Grassmannian, $\cW_{k,cn}$. Furthermore, we explicitly show that this bundle is non-orientable in the majority of the cases, and conjecture that it is non-orientable in the remaining situation. Using the combinatorics of the Deodhar decomposition of the Grassmannian, we identify subspaces $\Sigma(W) \subset \cW_{k,n}$ for which the restricted bundle lies outside the positive Grassmannian. Finally, while probing the combinatorics of the Deodhar decomposition, we give a diagrammatic algorithm for reading equations determining each Deodhar component as a semialgebraic set.
\end{abstract}

In this paper, we are interested in understanding the geometry represented by Wilson loop diagrams, which may be thought of as Feynman diagrams for SYM N=4 theory in twistor space \cite{Adamo:2011cb, Chicherinetal2016}. In recent years, there has been an active program to understand the scattering amplitudes of this theory geometrically \cite{Arkani-Hamed:2013jha, AmplituhedronDecomposition, wilsonloop, Amplituhedronsquared}. Namely, in \cite{Arkani-Hamed:2013jha}, the authors show that the on shell amplitudes of this theory correspond to the volume of a geometric space called an Amplituhedron. In \cite{Amplituhedronsquared}, the authors attempt to relate the geometry of the entire amplitude to the Amplituhedron. This paper concerns itself with some of the difficulties encountered in the latter attempt. In particular, we find that the space define by the tree level Feynman diagrams for SYM N=4 theory is non-orientable in many, if not all cases. This finding is consistent with the issues raised in \cite{Amplituhedronsquared} and \cite{HeslopStewart}.

We wish to emphasize that, while non-orientable spaces do not have a natural volume form, we do not believe that our findings in this paper pose a threat to the program of geometrically understanding the Wilson loop amplitudes. It is quite possible that the integrals associated to Wilson loop diagrams, \cite{Adamo:2011cb}, correspond to something far more subtle, such as characteristic classes of the prescribed geometric object, and that the Amplituhedron is only a special case of this phenomenon.

Much of this current work is based off of \cite{wilsonloop}, where the authors show that each Wilsoan loop diagram with $k$ propagators and $n$ vertices defines a subspace of the positive Grassmanian, $\Grnn(k, n)$, defined by points in $\Gr(k,n)$ whose Plu\"{u}cker coordinates are all non-negative. In \cite{generalcombinatorics}, the authors show that the subspace defined by each Wilson loop diagram is $3k$ dimensional. In \cite{casestudy}, the authors explicitly list all such subspaces defined by Wilson loop diagrams with $2$ propagators and $6$ vertices, as well as how they share boundaries with each other. The literature on the positive part of the Wilson loop diagram, while not complete, is coherent and clear.

However, the geometry of Wilson loop diagrams is not restricted to the positive Grassmannian. In \cite{Mason:2010yk}, the authors associate to each Wilson loop diagram the span of a family of $k$ vectors in $\R^{n+1}$ that need not represent an element of $\Grnn(k,n+1)$. Each of the $k$ vectors in this family is parametrized by $4$ independent coefficients. That is, each Wilson loop diagram defines a $4k$ parameter subspace of $\Gr(k,n+1)$ that is not in general contained in $\Grnn(k,n+1)$. The subspace of $\Grnn(k,n)$ defined by each Wilson loop diagram in \cite{wilsonloop} comes from a projection of each of these $k$ vectors onto $\R^n$. The physical quantity associated to the Wilson loop diagrams, the (tree level) scattering amplitude, is given by a sum of integrals, one for each Wilson loop diagram. One may view each integral as a volume form on the subspace of $\Grnn(k,n)$ defined by the Wilson loop diagram \cite{Amplituhedronsquared}. While each such integral is well defined, there are inconsistencies and problems when geometrically interpreting the sum of these integrals to get scattering amplitude \cite{HeslopStewart}. In this paper, we show that these problems and inconsistencies arise because the $4k$ dimensional subspace of $\Grnn(k,n+1)$ parametrized by the Wilson loop diagrams is not orientable. Therefore, it cannot have a global volume form.

In order to do this analysis, we cannot solely rely on the well understood CW-complex structure of $\Grnn(k,n)$, written in terms of positroids, Le-diagrams, Grassmann necklaces and  other cryptomorphic combinatorial tools. The positroid cell structure of $\Grnn(k,n)$ can be extended to a stratification of $\Gr(k, n)$. However, this is too coarse for our needs. To fully appreciate the geometry of Wilson loop diagrams, one must consider a decomposition of $\Gr(k,n)$ that is finer than the positroid stratification, namely, we discuss the Deodhar decomposition of $\Gr(k,n)$. This decomposition agrees with the positroid stratification on $\Grnn(k,n)$ and refines the positroid decomposition away from the positive part of $\Gr(k,n)$.

Section \ref{Deodhar strata section} introduces Deodhar components and the Go-diagrams that index them. Deodhar components are semialgebraic sets of $\Gr(k,n)$. In Theorem \ref{thm:sets_from_networks} we give a way of reading the defining equations of each Deodhar component from a certain network associated to each Go-diagram. Section \ref{WLD geometry section} introduces Wilson loop diagrams and the spaces they parametrize. Section \ref{section:fibers} describes the subset of $\Gr(k, n+1)$ parametrized by each Wilson loop diagram, $W$, as a fiber, $\pi^{-1}(\Sigma(W))$, of a certain projection over the positroid cell, $\Sigma(W)$, associated to $W$. We show that in general, the fiber of $\pi$ over a positroid cell may be written as a union of Deodhar components and describe the boundary structure of Deodhar components in this fiber. In particular the $\pi^{-1}(\Sigma(W))$ need not respect positivity. We give a combinatorial condition for when positivity is violated in Theorem \ref{thm:when C star not pos}. The fiber $\pi^{-1}(\Sigma(W))$ contains a unique top dimensional Deodhar component (Theorem \ref{thm:deodhar_components_of_fiber}), which should be thought of as playing the same role in $\Gr(k, n+1)$ that $\Sigma(W)$ plays in $\Grnn(k,n)$. Finally, we show that the space parametrized with Wilson loop diagrams is not orientable in Theorem \ref{thm:non-orientable}.

\section{Deodhar Strata \label{Deodhar strata section}}

In this section, we discuss a decomposition of the Grassmannian called the Deodhar decomposition, and the Go-diagrams that index them. Section \ref{Basic notions section} begins with some notation for the section. Section \ref{Go and Le section} introduces the positroid and Deodhar decompositions, as well as the Le and Go-diagrams that index them. As pre-existing knowledge about positroid strata and the positroid cell decomposition of the non-negative Grassmannian, $\Grnn(k,n)$, is not strictly necessary for this section, detailed background in this area is not explicitly provided in this paper. Finally, Section \ref{Network section} describes network parametrizations of Deodhar components, which is another combinatorial tool for understanding Deodhar components built from Go-diagrams. We show how the equations defining Deodhar components as semialgebraic subsets of $\Gr(k, n)$ may be read from these networks.

The results reviewed and presented in this section, while useful in their own right for understanding the geometry and stratification of generic Grassmanians, are also useful for understanding the subspace of $\Grnn(k,n)$ and $\Gr(k,n+1)$ parametrized by the Wilson loop diagrams. This is discussed in Section \ref{WLD geometry section}.

\subsection{Basic Notions \label{Basic notions section}}
We begin by fixing some notation. For a more general introduction to the correspondence between quotients of Weyl groups and flag varieties, see Chapters 2 and 3 of \cite{billey:singular_loci} or Part III of \cite{fulton:young_tableaux}. For an introduction to Deodhar decompositions of the Grassmannian see Chapter 5 of \cite{kodama:book}. For an introduction to positivity in the Grassmannian, see \cite{postnikov:total_positivity} or the lecture notes \cite{morales:notes}.

Let $s_i$ denote the Coxeter generator $(i,i+1)$ in the symmetric group $\mathfrak{S}_n$. We will use italicized lowercase letters, $v$, for permutations and bold faced letters, ${\textbf{v}}$, for specific expressions of permutations in the Coxeter generators. Further, if $v \in \fS_n$, at times, we write the explicit result of the permutation as $v = v(1)v(2) \ldots v(n)$. A {\textit{subexpression}} of $\textbf{v}$ is a permutation expressed in the Coxeter generators obtained by replacing some of the factors in ${\textbf{v}}$ by $\varepsilon$, the identity permutation in $\mathfrak{S}_n$. We use the terminology ``expression" and ``word" interchangeably.

Given an expression ${\textbf{v}} = v_1 v_2 \cdots v_m$ in the Coxeter generators, let $v_{(i)} = v_1 v_2 \cdots v_{i}$ be the product of the initial $i$ factors of ${\textbf{v}}$. So, $v_{(0)} = \varepsilon$ and $v_{(m)} = v$. Note the difference in this notation from the explicit expression of the permutation $v$.

The {\textit{length}} of a permutation $v$, denoted $\ell(v)$, is the minimum number of letters needed to write any expression of $v$. An expression is {\textit{reduced}} if $\ell\left(v_{(i+1)}\right) = \ell\left(v_{(i)}\right) + 1$ for each $i$. That is, the length of each permutation strictly increases as each subsequent letter in the word is included. All reduced expressions for a permutation contain the same number of factors. Equivalently, the length is the number of pairs $i < j$ such that $v(i) > v(j)$.
A subword $\textbf{u}$ of $\textbf{v}$ is {\textit{distinguished}} if whenever $\ell(u_{(i)}v_{i+1}) < \ell(u_{(i)})$, one also has $u_{i+1} = v_{i+1}$, (i.e. $u_{i+1} \neq \varepsilon$). A subexpression ${\textbf{u}}$ of ${\textbf{v}}$ is {\textit{positive}} if additionally $\ell\left(u_{(i+1)}\right) \geq \ell\left(u_{(i)}\right)$ for all $i$. In other words, a subword is positive if it is both distinguished and reduced (when ignoring the letters $\varepsilon$ appearing in $\textbf{u}$).

\begin{eg}
Consider the word \bas {\textbf{v}} = s_1 s_2 s_1 \;.\eas The word $\textbf{v}$ defines a permutation $v \in \fS_4$, it is written $ v= 3 2 1 4 = v(1)v(2)v(3)v(4)$. It is a reduced word of length $3$. The permutation can be represented by another reduced word ${\textbf{v}}' = s_2 s_1 s_2$. We may write $u= \varepsilon$ as a subword of $\textbf{v}$ in two ways, $\textbf{u}_1 = \varepsilon\varepsilon\varepsilon$ and $\textbf{u}_2 = s_1\varepsilon s_1$, where both are distinguished, but only $\textbf{u}_1$ is also reduced (i.e. positive). Similarly, we may write $w = s_1$ as a subword of $\textbf{v}$ in two ways, $\textbf{w}_1 = \varepsilon \varepsilon s_1$ and $\textbf{w}_1 = s_1\varepsilon \varepsilon$, both of which are reduced, but only $\textbf{w}_2$ is also distinguished (i.e. positive).
\end{eg}

\begin{lem}[Lemma 3.5 in \cite{marsh:parametrizations}] \label{lem:unique_positive_expression}
Let $u \leq v$ be permutations and $\textbf{v}$ be a reduced expression for $v$. Then, there is a unique positive subexpression for $u$ in $\textbf{v}$.
\end{lem}

The Young subgroup $\mathfrak{S}_k \times \mathfrak{S}_{n-k} \subset \mathfrak{S}_n$ acts on a permutation $v(1)v(2)\dots v(n)$ by letting $\mathfrak{S}_k$ act on $v(1)v(2) \dots v(k)$ and letting $\mathfrak{S}_{n-k}$ act on $v(k+1) \dots v(n)$. Any coset in the quotient $\mathfrak{S}_n /  (\mathfrak{S}_k \times \mathfrak{S}_{n-k})$ has a unique representative of the form $i_1 i_2 \dots i_k j_1 j_2 \dots j_{n-k}$ where $i_1 < i_2 < \cdots < i_k$ and $j_1 < j_2 < \cdots < j_{n-k}$. These permutations have one descent in the $k^{th}$ position. Such permutations are called \textit{Grassmannian}. Cosets in $\mathfrak{S}_n /  (\mathfrak{S}_k \times \mathfrak{S}_{n-k})$ are in bijection with subsets in $\binom{[n]}{k}$, which index Schubert cells of $\Gr(k,n)$. There is a family of partial orders $\gale{i}$ ($\sgale{i}$) on $\binom{[n]}{k}$ called shifted (strict) Gale orderings. Namely, let  $<_i$ be the linear order on $[n]$, $i <_i i+1 <_i n <_i 1 \ldots <_i i-1$. Let $I = \{i_1, i_2, \dots,  i_k\}$ and $J = \{j_1, j_2, \dots, j_k\}$ with $i_1 < i_2 < \cdots < i_k$ and $j_1 < j_2 < \cdots < j_k$. Then $I \gale{i} J$ if and only if $i_m \leq_i j_m$ for all $1 \leq m\leq k$. For instance, $1356 \gale{1} 1456$ but not $1356 \gale{4} 1456$. The partial order $\gale{1}$ agrees with the order on Schubert cells given by containment of closures.

Grassmannian permutations with a descent in the $k^{th}$ position are also in bijection with {\textit{Ferrers shapes}}.  These are collections of boxes obtained by taking a lattice path from the Northeast to Southwest corner of a $(n-k) \times k$ rectangle, then taking all boxes Northwest of this lattice path. We label the steps of the lattice path from $1$ to $n$ starting at the Northeast corner. The lattice path associated to the Grassmann permutation $i_1 i_2 \dots i_k j_1 j_2 \dots j_{n-k}$ takes vertical steps $i_1, i_2, \dots, i_k$ and horizontal steps $j_1, j_2, \dots, j_{n-k}$. We label a Ferrers shape by $\lambda$.

For any given box $b$ in a Ferrers shape, write $b = (i,j)$, where $b$ is in the row with vertical step $i$ and the column with horizontal step $j$. Furthermore, let $b^{in}$ be the set of boxes weakly to the right and weakly below $b$ aside from $b$ itself. Beware that this convention differs from that usually found in the literature, where $b \in b^{in}$. We find that some cumbersome notation and awkward expressions can be avoided by declaring $b \notin b^{in}$. Additionally set $b^{out} = D \setminus (b^{in} \cup b)$. 

\subsection{Le and Go-diagrams \label{Go and Le section}}

We are now ready to define Le and Go-diagrams.

Given a Ferrers shape, $\lambda$, write the associated Grassmannian permutation $v_\lambda = i_1\ldots i_k j_1 \ldots, j_{n-k}$, where $i_m = v(m)$, $j_r = v(r+k)$. Associate to a Ferrers shape an expression $\textbf{v}$ by filling in the boxes of a Ferrers shape as follows:

\begin{enumerate}
\item Fill the top left (Northwest most) box with $s_{n-k}$.
\item For any box $b$ filled with the permutation $s_r$, fill the box to the right, if it exists, with $s_{r-1}$ and the box below, if it exists, with $s_{r+1}$.
\end{enumerate}

A {\textit{reading order}} on a Ferrers shape containing $m$ boxes is a filling of the boxes with the integers from $1$ to $m$ which is increasing upward and to the left. Reading the transpositions decorating the boxes of the Ferrers shpae in any reading order yields an expression $\textbf{v}$ for $v_\lambda$. 

A $\circ/+$-diagram is a Ferrers shape whose boxes have been filled with pluses, $+$'s, or white stones $\circ$'s. $\circ/+$-diagrams correspond to subexpressions $\textbf{u}$ of $\textbf{v}$ by filling the boxes containing transpositions used in$\textbf{u}$ with $\circ$'s and boxes containing transpositions not used in $\textbf{u}$ with $+$'s. 

\begin{thm}[Proposition 4.5 in \cite{lam:total_positivity}] \label{thm:reading_order_doesnt_matter}
Let $D$ be a $\circ/+$-diagram and let $\textbf{u}, \textbf{v}$ be the associated subword, word pair in some reading order.
\begin{enumerate}
\item The permutations $v$, coming from the Ferrers shape, and $u$, coming from the $\circ/+$-filling do not depend on the choice of reading order.
\item Whether $\textbf{u}$ is a distinguished subexpression of $\textbf{v}$ depends only $D$, not on the choice of reading order.
\end{enumerate}
\end{thm}

This Theorem is proved by noting that if $\textbf{u}$ and $\textbf{u}'$ obtained by changing the reading order on the same $\circ/+$-diagram, then they are related by commutations of the Coxeter generators, and thus $u = u'$.

A \textit{Le-diagram} is defined to be an $\circ/+$-diagram corresponding to a positive subexpression. From Lemma \ref{lem:unique_positive_expression}, there is a unique Le-diagram corresponding to any pair $u \leq v$ where $v$ is Grassmann. Le-diagrams are commonly defined as $\circ/+$-diagrams such that there is no $\circ$ with a $+$ to its left in its row and a $+$ above it in its column. This restriction on diagrams is called the {\textit{Le-property}}. These two definitions of Le-diagrams are shown to be equivalent in Theorem 5.1 in \cite{lam:total_positivity}.

Le-diagrams index positroid strata in $\Gr(k,n)$. When these are intersected with $\Grnn(k,n)$, it gives a cell complex structure to $\Grnn(k,n)$. The positroid cells correspond to matroids, and are in one to one correspondence with the positroid strata. The positroid strata are also in bijection with several other combinatorial objects including Grassmann necklaces. An algorithm to pass from a Le-diagram to a Grassmann Necklace is given in \cite{Ohsalg}. Given a Grassmann necklace $\cI =\{I_1, \ldots , I_n\}$, the associated positroid cell is the semialgebraic subset of $\Grnn(k,n)$ given by
\ba \label{eqn:positroid cell plucker}
\Delta_J = 0 \textrm{ if and only if } J \prec_{m} I_m
\textrm{ for each } m \in [n] ; \quad \textrm{ otherwise } \Delta_J > 0 \;.\ea  Here $\gale{m}$ indicates the $m^{th}$ shifted Gale order.   There are several ways to relate the positroid cell decompostion on $\Grnn(k,n)$ to a stratification on $\Gr(k,n)$. Here we have adopted the convention established in \cite{knutson:juggling}, that the same Grassmann Necklace is associated to a positroid strata given by the semialgebraic subset of $\Gr(k,n)$
\ba 
\Delta_{I_m} \neq 0\quad ; \quad
\Delta_J = 0 \textrm{ if } J \prec_{m} I_m
\textrm{ for each } m \in [n] \ea  No constraints are placed on the remaining Pl\"{u}cker coordinates. 

We now study $\bullet/\circ/+$ diagrams, a specific subset of which correspond to distinguished subexpressions. As with the $\circ/+$ diagrams, let $\textbf{v}$ be the word associated to the shape $\lambda$. Given a subword $\textbf{u}$, the boxes of the Ferrers diagram containing transpositions used in $\textbf{u}$ are filled with $\circ$'s and $\bullet$'s and boxes containing transpositions not used in $\textbf{u}$ are filled with $+$'s. A box is filled with a $\bullet$ if the transposition labelling the box is used in $\textbf{u}$ and its multiplication decreases the length of the word.

Given a box $b$, let $u_{b^{in}}^{D}$ be the permutation obtained by multiplying the transpositions in boxes containing stones in $b^{in}$ in $D$ a valid reading order. If the diagram is clear from context, we will just write $u_{b^{in}}$ instead of $u_{b^{in}}^D$. Let $s_b$ be the transposition in box $b$. The we define a $\bullet/\circ/+$ diagram to be a $\circ/+$-diagram where stones in boxes $b$ such that $\ell(u_{b^{in}} s_b) < \ell(u_{b^{in}})$ have been colored black, $\bullet$. In this manner, we can consider only the diagrams $D$ that correspond to a distinguised subword. We say that a box $b$ {\textit{violates the distinguished property}} if $\ell( u_{b^{in}}s_b) < \ell(u_{b^{in}})$, but $b$ is filled with a $+$. 

The word subword pair, $\textbf{u} \prec \textbf{v}$, associated to $D$ is distinguished if and only if $D$ does not have any boxes that violate the distinguished poperty. Such diagrams are called \emph{Go-diagrams}. As a corollary of Theorem \ref{thm:reading_order_doesnt_matter}, we see that if $D$ is a $\bullet/\circ/+$-diagram and $\textbf{u}, \textbf{v}$ is the associated subword, word pair in some reading order, then whether $\textbf{u}$ is a distinguished subexpression of $\textbf{v}$ depends only $D$, not on the choice of reading order. Go-diagrams index \emph{Deodhar components} in $\Gr(k,n)$. Restricted to $\Grnn(k,n)$, the Deodhar decomposition agrees with the positroid decomposition (as stated in Theorem \ref{thm:positive_deodhar_components} below). Away from the non-negative part of $\Gr(k,n)$, the Deodhar decomposition refines the positroid decomposition. 

Like a positroid cell, a Deodhar component is defined by setting certain Pl\"ucker coordinates to zero and demanding others do not vanish. Each square of a diagram corresponds to a constraint placed on a Pl\"{u}cker coordinate.

\begin{dfn} \label{def:set_I_b}
Let $D$ be a Go-diagram with Ferrers shape $\lambda$. For $J = j_1\ldots j_k \in \binom{[n]}{k}$, write $\proj_{J}(v)$ to be the set $\{v(j_1),\ldots, v(j_k)\} \in \binom{[n]}{k}$. Define $I_b = \proj_{I_\lambda}(u_{b^{in}} v_b (v_{b^{in}}))^{-1})$. \end{dfn}

$D$ prescribes constrains on the Pl\"{u}cker coordinates as follows:

\begin{dfn}\label{dfn:plucker_coordinates}
Let $D$ be a Go-diagram. The {\textbf{Deodhar component}} $\D$ is the subset of $\Gr(k,n)$ determined by the constraints:
\begin{displaymath}
\begin{split}
\Delta_{I_b} = 0 & \mbox{ if } b \mbox{ is filled by a } \circ \\
\Delta_{I_b} \neq 0 & \mbox{ if } b \mbox{ is filled by a } +.
\end{split}
\end{displaymath}
for each $b \in D$, and, if the Ferrers shape of $D$ is $\lambda$, the constraints $\Delta_{I_\lambda} \neq 0$, and $\Delta_{J} = 0$ if $J \prec_1 I_{\lambda}$.
\end{dfn}

Note that a reduced distinguished subword is positive. So, a Go-diagram with no $\bullet$'s is a Le-diagram. This observation is reflected geometrically.

\begin{thm}[Theorem 5.13 in \cite{kodama:deodhar_decomposition}] \label{thm:positive_deodhar_components}
Let $D$ be a Go-diagram and $\mathcal{D}$ the associated Deodhar component.
\begin{enumerate}
\item $D$ contains no $\bullet$'s if and only if $D$ is a Le-diagram.
\item $\mathcal{D} \cap \Grnn(k,n) \neq \emptyset$ if and only if $D$ is a Le-diagram. In this case $\mathcal{D} \cap \Grnn(k,n)$ is identical to the restriction of the positroid strata labelled by $D$ to $\Grnn(k,n)$.
\end{enumerate}
\end{thm}

The Deodhar decomposition also refines the Richardson stratification of the Grassmannian. While the details of Richardson strata, the intersection of Schubert cells with opposite Schubert cells, is not necessary for this paper, we will make use of the following consequence of this fact.

\begin{prop}[Corollary 1.2 in \cite{deodhar:geometric_aspects}] \label{prop:refine_richardson}
Let $\cD$ be the Deodhar component associated to the distinguished subword pair $\textbf{u} \prec \textbf{v}$. Let $S = \{u(n), u(n-1), \ldots, u(n-k)\}$ Then, $\Delta_{S} \neq 0$ on $\cD$ and $\Delta_{J} = 0$ on $\cD$ for all $S \prec_1 J$.
\end{prop}

\subsection{Network parametrizations \label{Network section}}

This section introduces a parametrization of Deodhar components defined in \cite{talaska:network_parametrizations} and give a means of reading off the sets $I_b$ from Definition \ref{def:set_I_b} via a graphical algorithm.

\begin{dfn}[Definition 3.2 in \cite{talaska:network_parametrizations}] \label{def:go-network}
The {\textit{Go-network}}, $N(D)$, associated to a Go-diagram $D$ is built by:
\begin{itemize}
\item Placing a {\textit{boundary vertex}} along each edge of $D$'s southeast border.
\item Placing an {\textit{internal vertex}} for each $+$ or $\bullet$ in the diagram. We call these vertices $+$-vertices and $\bullet$-vertices.
\item From each internal vertex, drawing an edge right to the nearest $+$-vertex or boundary vertex.
\item From each internal vertex, drawing an edge down to the nearest $+$-vertex or boundary vertex.
\item Directing all edges left and down.
\end{itemize}
\end{dfn}

The vertical steps of $D$'s boundary become sources in the Go-network and the horizontal steps become sinks.

\begin{eg} \label{ex:go-network}
The following is an example of a Go-diagram and its associated Go-network. The $\bullet$-vertices have been drawn at a larger size to distinguish them.

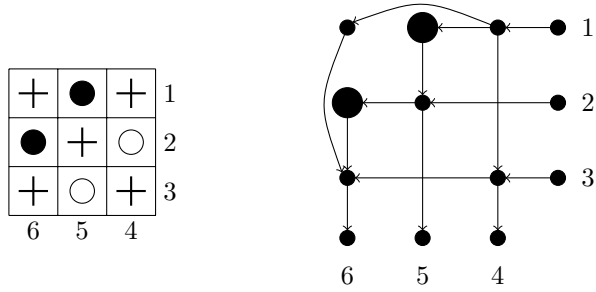
\begin{displaymath}
\begin{tikzpicture}
\begin{scope}[scale=.65]
\draw[step = 1] (0,0) grid (3,3);
\draw (3.3,2.5) node {$1$};
\draw (3.3,1.5) node {$2$};
\draw (3.3,0.5) node {$3$};
\draw (2.5,-.3) node {$4$};
\draw (1.5,-.3) node {$5$};
\draw (0.5,-.3) node {$6$};
%
\draw[thick] (0.5,0.2) -- (0.5,0.8);
\draw[thick] (0.2,0.5) -- (0.8,.5);
\draw[fill = white] (1.5,0.5) circle (.25);
\draw[thick] (2.5,0.2) -- (2.5,0.8);
\draw[thick] (2.2,0.5) -- (2.8,0.5);
%
\draw[fill = black] (0.5,1.5) circle (.25);
\draw[thick] (1.5,1.2) -- (1.5,1.8);
\draw[thick] (1.2,1.5) -- (1.8,1.5);
\draw[fill = white] (2.5,1.5) circle (.25);
%
\draw[thick] (0.5,2.2) -- (0.5,2.8);
\draw[thick] (0.2,2.5) -- (0.8,2.5);
\draw[fill = black] (1.5,2.5) circle (.25);
\draw[thick] (2.5,2.2) -- (2.5,2.8);
\draw[thick] (2.2,2.5) -- (2.8,2.5);
\end{scope}

\begin{scope}[shift={(4,0)}]
\draw[fill = black] (3.3,2.5) circle (.1);
\draw (3.7,2.5) node {$1$};
\draw[fill = black] (3.3,1.5) circle (.1);
\draw (3.7,1.5) node {$2$};
\draw[fill = black] (3.3,0.5) circle (.1);
\draw (3.7,0.5) node {$3$};
\draw[fill = black] (2.5,-.3) circle (.1);
\draw (2.5,-.8) node {$4$};
\draw[fill = black] (1.5,-.3) circle (.1);
\draw (1.5,-.8) node {$5$};
\draw[fill = black] (0.5,-.3) circle (.1);
\draw (0.5,-.8) node {$6$};
%
\draw[fill = black] (0.5,0.5) circle (.1);
\draw[fill = black] (2.5,0.5) circle (.1);
%
\draw[fill = black] (0.5,1.5) circle (.2);
\draw[fill = black] (1.5,1.5) circle (.1);
%
\draw[fill = black] (0.5,2.5) circle (.1);
\draw[fill = black] (1.5,2.5) circle (.2);
\draw[fill = black] (2.5,2.5) circle (.1);
%
%
\draw[->] (3.3,0.5) -- (2.6,0.5);
\draw[->] (2.5,0.5) -- (0.6,0.5);
\draw[->] (2.5,0.5) -- (2.5,-.2);
\draw[->] (0.5,0.5) -- (0.5,-.2);
%
\draw[->] (3.3,1.5) -- (1.6,1.5);
\draw[->] (1.5,1.5) -- (0.7,1.5);
\draw[->] (1.5,1.5) -- (1.5,-.2);
\draw[->] (0.5,1.5) -- (0.5,.6);
%
\draw[->] (3.3,2.5) -- (2.6,2.5);
\draw[->] (2.5,2.5) -- (1.7,2.5);
\draw[->] (2.5,2.5) -- (2.5,.6);
\draw[->] (1.5,2.5) -- (1.5,1.6);
\draw[->] (2.5,2.5) .. controls (1.5,2.9) .. (0.5707,2.5707);
\draw[->] (0.5,2.5) .. controls (.1,1.5) .. (0.4293,0.5707);
\end{scope}

\end{tikzpicture}
\end{displaymath}
\end{eg}

\begin{dfn}
A {\textit{weighted Go-network}} is a Go-network where the edges directed left into $+$-vertices are weight by elements of $\rr^{\ast}$ and the edges directed left into $\bullet$-vertices are weighted by elements of $\rr$. To a weighted Go-network, we associate the set of coordinates
\begin{equation} \label{eqn:network_coordinates}
\left\{
\Delta_J = \sum_{P: I \to J} \mathrm{sgn}(P) \prod_{p \in P} \prod_{e \in p} w(e) : J \subseteq [n]
\right\},
\end{equation}
\noindent
where sum is across all collections $P$ of vertex disjoint paths from the set of sources $I$ of the Go-network to $J$. The product is the product of the weights of all the edges appearing in all paths in $P$. The sign $\mathrm{sgn}(P)$ is the sign of $P$ viewed as a partial permutation. That is, $\mathrm{sgn}(P) = (-1)^{c}$, where $c$ is the number of edge crossings among the paths in $P$.
\end{dfn}

\begin{eg}
Consider the following weighting of the Go-network from Example \ref{ex:go-network}.
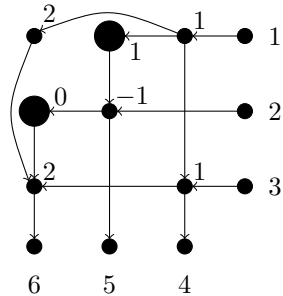
\begin{displaymath}
\begin{tikzpicture}
\begin{scope}[scale=1]
\draw[fill = black] (3.3,2.5) circle (.1);
\draw (3.7,2.5) node {$1$};
\draw[fill = black] (3.3,1.5) circle (.1);
\draw (3.7,1.5) node {$2$};
\draw[fill = black] (3.3,0.5) circle (.1);
\draw (3.7,0.5) node {$3$};
\draw[fill = black] (2.5,-.3) circle (.1);
\draw (2.5,-.8) node {$4$};
\draw[fill = black] (1.5,-.3) circle (.1);
\draw (1.5,-.8) node {$5$};
\draw[fill = black] (0.5,-.3) circle (.1);
\draw (0.5,-.8) node {$6$};
%
\draw[fill = black] (0.5,0.5) circle (.1);
\draw[fill = black] (2.5,0.5) circle (.1);
%
\draw[fill = black] (0.5,1.5) circle (.2);
\draw[fill = black] (1.5,1.5) circle (.1);
%
\draw[fill = black] (0.5,2.5) circle (.1);
\draw[fill = black] (1.5,2.5) circle (.2);
\draw[fill = black] (2.5,2.5) circle (.1);
%
%
\draw[->] (3.3,0.5) -- (2.6,0.5);
\draw[->] (2.5,0.5) -- (0.6,0.5);
\draw[->] (2.5,0.5) -- (2.5,-.2);
\draw[->] (0.5,0.5) -- (0.5,-.2);
%
\draw[->] (3.3,1.5) -- (1.6,1.5);
\draw[->] (1.5,1.5) -- (0.7,1.5);
\draw[->] (1.5,1.5) -- (1.5,-.2);
\draw[->] (0.5,1.5) -- (0.5,.6);
%
\draw[->] (3.3,2.5) -- (2.6,2.5);
\draw[->] (2.5,2.5) -- (1.7,2.5);
\draw[->] (2.5,2.5) -- (2.5,.6);
\draw[->] (1.5,2.5) -- (1.5,1.6);
\draw[->] (2.5,2.5) .. controls (1.5,2.9) .. (0.5707,2.5707);
\draw[->] (0.5,2.5) .. controls (.1,1.5) .. (0.4293,0.5707);
%
\draw (.7,2.8) node {$2$};
\draw (1.85,2.3) node {$1$};
\draw (2.7,2.7) node {$1$};
%
\draw (1.8,1.7) node {$-1$};
\draw (.85,1.7) node {$0$};
%
\draw (2.7,0.7) node {$1$};
\draw (0.7,0.7) node {$2$};
\end{scope}
\end{tikzpicture}
\end{displaymath}
\noindent
The only collection of paths from sources $\{1,2,3\}$ to $\{1,2,3\}$ is the empty collection of paths, so $\Delta_{123} = 1$. From $\{1,2,3\}$ to $\{2,3,6\}$, we need a path from $1$ to $6$. There are two such paths: the one along the top and left of the network has weight $2$ and the one throught the middle of the network has weight $0$. So, $\Delta_{236} = 2$. To $\{3,4,5\}$, there is one collection of vertex disjoint paths: a path from $1$ to $4$ of weight $1$ and a path from $2$ to $5$ of weight $-1$. These paths have one edge crossing, introducing a factor of $-1$. So, $\Delta_{345} =1$. There is no collection of vertex disjoint paths to $\{4,5,6\}$, so $\Delta_{456} = 0$. Continuing in this way, one may compute:
\begin{displaymath}
\begin{array}{rrrrrrr}
\Delta_{123} & = & 1, & \qquad & \Delta_{234} & = & 1, \\
\Delta_{124} & = & 1, & & \Delta_{235} & = & 1, \\
\Delta_{125} & = & 0, & & \Delta_{236} & = & 2, \\
\Delta_{126} & = & 2, & & \Delta_{245} & = & 1, \\
\Delta_{134} & = & 0, & & \Delta_{246} & = & 2, \\
\Delta_{135} & = & -1, & & \Delta_{256} & = & -2, \\
\Delta_{136} & = & 0, & & \Delta_{345} & = & 1, \\
\Delta_{145} & = & -1, & & \Delta_{346} & = & 0, \\
\Delta_{146} & = & 0, & & \Delta_{356} & = & -2, \\
\Delta_{156} & = & 2, & & \Delta_{456} & = & -2. \\
\end{array}
\end{displaymath}
\end{eg}

\begin{thm}[Theorem 3.16 in \cite{talaska:network_parametrizations}] \label{thm:go-networks}
The set of coordinates (\ref{eqn:network_coordinates}) of a weighted Go-network of shape $N(D)$ is the set of Pl\"ucker coordinates of some point in the Deodhar component $\cD$. The map from weighted Go-networks to their coordinates is a bijection between weighted Go-networks of shape $N(D)$ and points in the Deodhar component $\cD$.
\end{thm}

\begin{cor} \label{cor:dim_of_deodhar_component}
Let $\cD$ be the Deodhar component labelled by the Go-diagram $D$. Then,
\begin{displaymath}
\cD \simeq \rr^{\# \ \mathrm{of} \ \bullet'\mathrm{s} \ \mathrm{in} \ D} \times (\rr^{\ast})^{\# \ \mathrm{of} \ +'\mathrm{s} \ \mathrm{in} \ D}.
\end{displaymath}
\end{cor}

\begin{cor} \label{cor:no_flow_vanishes}
Let $D$ be a Go-diagram indexing the Deodhar component $\cD$. The Pl\"ucker coordinate $\Delta_S$ vanishes uniformly on $\cD$ if and only if there is no flow to $S$ in $N(D)$.
\end{cor}

\begin{proof}
If there is not a flow to $S$, the sum (\ref{eqn:network_coordinates}) is empty, so $\Delta_S$ vanishes uniformly on $\cD$. If there is a flow to $S$, choosing all of the edge weights to be algebraically independent yields a point in $\cD$ where $\Delta_S \neq 0$.
\end{proof}

We will make reference to the following special case of this corollary.

\begin{cor} \label{cor:white_stone_no_flow}
Let $D$ be a Go-diagram and suppose $b \in D$ contains a white stone. Then, there is not a flow to $I_b$ in $N(D)$.
\end{cor}

For a point $p$ in the Deodhar component $\cD$ and set $S \in \binom{[n]}{k}$, the corresponding Pl\"ucker coordinate, $\Delta_S(p)$, of $p$ will not, in general, depend on the entire network. For instance, if $i_1 \in S$, then $i_1$ will not be involved in any flow determining $\Delta_S(p)$. That is, the value of $\Delta_S(p)$ will not depend on the top row of the network. Similarly, if $j_n \notin S$, the value of $\Delta_S(p)$ will not depend on the leftmost column of the network. This generalizes to the following observation.

\begin{cor} \label{cor:restricting_to_columns}
Let $D$ be a Go-diagram whose Ferrers shape has vertical steps $I = \{i_1, \dots, i_k\}$ and horizontal steps $J = \{j_1, \dots, j_{n-k}\}$ with $i_1 < \cdots < i_k$ and $j_1 < \cdots < j_{n-k}$. Let $p$ lie in the Deodhar component indexed by $D$ let its associated weighted Go-network be $N(D)(p)$. For $S \in \binom{[n]}{k}$ suppose the first $\ell$ elements of $I$ and first $m$ elements of $J$ are not in $S$. That is $\{i_1, i_2, \dots, i_{\ell}\} \subseteq S$ and $\{j_n, \dots, j_{n-m+1}\} \not \subset S \cap J$. Then, the Pl\"{u}cker coordinate $\Delta_{S}(p)$ is determined by the restriction of $N(D)$ to the rows $i_{\ell+1}, i_{\ell + 2}, \dots, i_{k}$ and columns $j_1, \dots, j_{n-m}$ of $N(D)(p)$.
\end{cor} 

The network perspective gives an alternate algorithm for reading off the sets $I_b$ defining equations for the Deodhar component from Definition \ref{dfn:plucker_coordinates}. Suppose $b = (i_\ell, j_m)$ is in the row with vertical step $i_\ell$ and column with horizontal step $j_m$. The proof relies on interpreting $S = I_b \setminus j_m \cup i_\ell$ as the maximal element of $\binom{[n]}{k}$ in the $\gale{1}$ ordering that satisfies $i_1, i_2, \dots, i_{\ell} \in S$, $j_n, j_{n-1}, \dots, j_{m} \notin S$, and such that $\Delta_S$ does not vanish uniformly on $\cD$.

\begin{thm} \label{thm:sets_from_networks}
Let $b$ be a box in a Go-diagram $D$. Then, $I_b \in \binom{[n]}{k}$ is the maximal set in the $\gale{1}$ ordering such that there is a collection of vertex disjoint paths flowing from the source nodes to $I_b$ in the network associated to the diagram $D'$ obtained by:
\begin{itemize}
\item changing the filling of $b$ to a plus,
\item changing the filling of all boxes in $b^{out}$ to white stones, and
\item changing the filling of all boxes in $b$'s row and $b$'s column aside from $b$ to white stones.
\end{itemize}
\end{thm}

\begin{proof}
Let $D$ be a Go-diagram whose Ferrers shape has vertical steps $i_1 < \cdots < i_k$ and horizontal steps $j_1 < \cdots < j_{n-k}$. Let $b = (i_{\ell},j_m)$ be a box in the Go-diagram $D$, and let $D'$ be the diagram described in the theorem statement. Let $J_b\in \binom{[n]}{k}$ is the maximal set in the $\gale{1}$ ordering such that there is a collection of vertex disjoint paths flowing from the source nodes to $I_b$ in the network $N(D')$. Let $J'_b = J_b \setminus j_{m} \cup i_{\ell}$. Then, $J'_b$ is the maximal set in $\gale{1}$ such that $\Delta_S$ does not vanish uniformly on $\cD$, subject to the constraints
\begin{equation} \label{eqn:deletion_contraction_constraints}
\begin{split}
i_1, i_2, \dots, i_{\ell} \in J'_b, \mbox{ and} \\
j_n, j_{n-1}, \dots, j_m \notin J'_b.
\end{split}
\end{equation}
We show that $I_b = J'_b \setminus i_{\ell} \cup j_{m}$.

Let $E$ be the diagram obtained by changing the filling of $b$ to a white stone in $D'$. Note that $E$ is still a Go-diagram. 
Let $\cE$ be the Deodhar component determined by $E$. The Go-network of $E$ has no nodes in the rows $i_1, i_2, \dots, i_\ell$ or the columns $j_n, j_{n-1}, \dots, j_m$. So, Theorem \ref{thm:go-networks} implies that $\Delta_{S}$ vanishes on $\cE$ if $\{i_1, i_2, \dots, i_\ell\}\nsubseteq S$ or $S \cap \{j_n, j_{n-1}, \dots, j_m\} \neq \emptyset$. Let $\cV \subseteq Gr(k,n)$ be the variety defined by
\begin{displaymath}
\Delta_S = 0 \mbox{ for all $S$ such that $\{i_1, i_2, \dots, i_\ell\}\nsubseteq S$ or $S \cap \{j_n, j_{n-1}, \dots, j_m\} \neq \emptyset$}.
\end{displaymath}
\noindent
Corollary \ref{cor:restricting_to_columns} then implies that $\cE = \cD \cap \cV$. Let $u',v'$ be the permutations associated to the diagram $E$. Proposition \ref{prop:refine_richardson} says that $\Delta_{u'(n)u'(n-1)\ldots u'(n-k)} \neq 0$ on $\cE$ and $\Delta_S = 0$ on $\cE$ for all $S > \{u(n), u(n-1), \ldots, u(n-k)\}$. So,
\begin{displaymath}
\{u(n), u(n-1), \ldots, u(n-k)\} = J'_b.
\end{displaymath}
From Definition \ref{def:set_I_b}, $J'_b = \proj_{I_\lambda}(s_b u^{E}_{b^{in}} v^{E}_b (v^{E}_{b^{in}}))^{-1})$.
Turning our attention back to the diagram $D'$, the transposition $s_b$ in box $b$ exchanges $i_{\ell}$ and $j_m$. So, $I_b = J'_b \setminus i_{\ell} \cup j_m = J_b$, as desired.
\end{proof}

\section{Geometry of Wilson Loop Diagrams\label{WLD geometry section}}

Having discussed the Deodhar decomposition of $\Gr(k,n)$, we apply this information to understand the geometric structure of Wilson loop diagrams. In Section \ref{WLD defs section}, we introduce the combinatorics of Wilson loop diagrams. We do not, in the paper, discuss the integrals associated to Wilson loop diagrams or their physical significance. We are only interested in the geometric subspaces they parametrize. For the physical significance of these diagrams, see, for instance, \cite{Mason:2010yk, Adamo:2011cb, Amplituhedronsquared, Chicherinetal2016}.

In \cite{casestudy}, the authors introduced a subspace $\cW_{k,n} \subset \Grnn(k,n)$ that is parameterized by the Wilson loop diagrams with $n$ vertices and $k$ propagators. In Section \ref{section:fibers}, we identify a subspace $\Omega_{n+1} \subset \Gr(k,n+1)$ such that $\pi : \Omega_{n+1} \rightarrow \Grnn(k,n)$ is a vector bundle. We show that the restriction of this bundle to a positroid cell is a union of Deodhar components in $\Gr(k,n+1)$. In Section \ref{bundle section}, we restrict this bundle to $\cW_{k,n}$ and show that this is exactly the subset of $\Gr(k,n+1)$ obtained by incorporating the gauge vector into the geometry of Wilson loop diagrams. In Section \ref{boundary chasing section}, we discuss the boundary structure of the cell decomposition of $\cW_{k,n}$, and introduce two families of Wilson loop diagrams that exhibit a curious geometric pattern \cite{Heslopcommunication}. In Section \ref{orientation section}, we use these families of Wilson loop diagrams to show that the subspace $\pi^{-1} (\cW_{k,n}) \subset \Gr(k,n+1)$, defined by the introduction of the gauge vector, is not orientable.

\subsection{Wilson Loop diagrams \label{WLD defs section}}

In this section, we turn to the geometry of Wilson loop diagrams. As combinatorial objects, we present these as the following data:

\begin{dfn}
Let $[n] = \{1, \ldots , n\}$, cyclically ordered. A Wilson Loop diagram, $W$,  is defined by a set of ordered pairs $\cP \subset [n] \times [n]$. We denote it $W = (\cP, n)$.
\end{dfn}

Graphically, we denote this by a convex polygon, with vertices labeled by the elements of $[n]$, respecting the cyclic ordering. For each $p = (i_p, j_p) \in \cP$, draw an internal wavy line between the edges of the polygon defined by the vertices $\{i_p, i_p +1\}$ and $\{j_p, j_p +1\}$. The set $\cP$ is called the set of propagators. We use the convention that $i_p <_1 j_p$.

\begin{dfn}
Let $W = (\cP, n)$ be a Wilson loop diagram.
\begin{enumerate}
\item The support of a set of propagators is the collection of vertices that define the endpoints of the edges the propagators land on. For $p = (i_p, j_p) \in \cP$ , the support of $p$ is written $V(p) = \{i_p, i_p+1, j_p, j_p+1\}$. For $P \subset \cP$, the support is written $V(P) = \bigcup_{q \in P} V(q)$.
\item The propagator set of a set of vertex $v$ is $\{p \in \cP | v \in V(p)\}$. For $V \subset [n]$, $\Prop(V) = \{p \in \cP| V(p) \cap V \neq \emptyset \}$.
\end{enumerate}
\end{dfn}

Geometrically, Wilson loop diagrams parametrize subspaces of $\Gr (k, n)$. Define a matrix of indeterminates associated to each Wilson loop diagram as follows.

\begin{dfn} \label{dfn C}
Let $W = (\cP, n)$ be a Wilson loop diagram. Impose an order on the set of propagators. Define $C(W)$ to be a variable valued matrix with entries \bas C(W)_{p,q} = \begin{cases} c_{p,q} & \textrm{ if } q \in V(p) \\
0  & \textrm{ if } q \not \in V(p).  \end{cases}
\; \eas Here, the entries $c_{p,q}$ are real variables.
\end{dfn}

The variables $c_{p,q}$ are taken to be algebraically independent variables. We interpret $C(W)$ to parametrize set of points in $\Grnn(k,n)$ defined by setting $\Delta_I = 0$ whenever associated minor of $C(W)$ vanishes (when $\Delta_I(C(W)) = 0$) and positive otherwise (when $\Delta_I(C(W)) \neq 0$).  The minor $\Delta_I(C(W))$ is taken in the ring $\rr[c_{p,q}]$. Equivalently, $C(W)$ parametrizes the set of points $\Grnn(k,n)$ which are row spans of full rank matrices obtained by evaluating the entries $c_{p,q}$ of $C(W)$ at real numbers. 

There are two clarifying remarks to be made here. First, since the intersection of positroid strata and $\Grnn(k,n)$ defines a matroid, the space parametrized by $C(W)$ is defined by all the minors of $C(W)$, not just the $0$ minors. Second, there are different conventions for the entries of $\Grnn(k,n)$ such that $C(W)$ parametrizes the subspace of $\Gr(k,n)$ that are given as row spans of full rank matrices obtained by evaluating the entries $c_{p,q}$ of $C(W)$ at positive real numbers. However, this representative is visually more complicated, so we do not use it here.

\begin{dfn} \label{dfn C star}
We define an augmented matrix, $C_*(W)$ by adjoining a column to $C(W)$. \bas C_*(W)_{p,q} = \begin{cases} c_{p,q} & \textrm{ if } q \in V(p) \cup n+1 \\
0  & \textrm{ if } q \not \in (V(p) \cup n+1) \end{cases}
\;. \eas \end{dfn}

We interpret $C_*(W)$ as a matrix parametrizing the set of points in $\Gr(k,n+1)$ determined by setting $\Delta_I = 0$ if and only if the associated minor of $C_{*}(W)$ is $0$, $\Delta_I(C_*(W)) = 0$. Otherwise, if $n+1 \not \in I$, the associated minor is positive. If $n+1 \in I$, we place no restriction on the value of the corresponding Pl\"{u}cker coordinate. So, $C_*(W)$ is defined by the same set of equations defining $C(W)$ viewed as equations on $\Gr(k,n+1)$ rather than $\Grnn(k,n)$. 




We are interested in a particular class of Wilson loop diagrams called admissible Wilson loop diagrams.

\begin{dfn}
A Wilson loop diagram $W = (\cP, n)$ is admissible if the following hold:
\begin{enumerate}
\item $n \geq |\cP|+4$
\item $\forall P \subset \cP$, $|P| + 3 \leq |V(P)|$
\item if $p = (i_p, j_p)$, $q = (i_q, j_q) \in \cP$ are two propagators, then $i_p<i_q<j_q<j_p$ in the cyclic ordering of $[n]$
\end{enumerate}
\end{dfn}

The first condition ensures that there are at least $4$ more vertices than propagators. The second ensures that each subset of propagators is supported on at least $3$ more vertices than the set of propagators. In particular, a propagator cannot start and end on the same edge, or on an adjacent edge. Nor can two propagators start and end on the same edges. The last condition ensures that there are no crossing propagators.

In \cite{wilsonloop}, the authors show that if $W$ is an admissible Wilson loop diagram, then $C(W)$ defines a matroid that is also a positroid. In other words, the space parametrized by $C(W)$ intersects $\Grnn(k,n)$. Furthermore, this intersection is a positroid cell, which we call $\Sigma(W)$.

Physically, the admissible Wilson loop diagrams define (tree level) particle interaction in SYM N=4 theory. For $W = (\cP, [n])$, the cyclically ordered set $[n]$ corresponds to the external particles associated to an interaction. The set $\cP$ corresponds to the $MHV$ propagators of the interaction. Each particle is represented as a section of a $|\cP|$-vector bundle over twistor space, projected onto a real subspace by a process called bosonization \cite{Arkani-Hamed:2013jha}. In other words, each vertex is labeled by a vector, $Z_i \in \R^{4 + |\cP|}$. Let $\cZ$ be the matrix whose $i^{th}$ row is the vector $Z_i$. The $Z_i$ labeling the vertices of $W$ satisfy $\cZ \in M_{\R, +}(n, |\cP|+ 4)$, the space of $n \times (|\cP|+4)$ matrices with nonnegative minors. Let $Z_* \in \R^{4+ |\cP|}$ be a gauge vector and indicate by $\cZ_*$ the matrix whose $i^{th}$ row is $Z_i$, if $i \leq n$, and the $n+1^{sh}$ row is $Z_*$. We place no restrictions on the positivity of $\cZ_*$. Given an appropriate choice of gauge vector, $Z_*$, one may associate an integral to the space parametrized by $C_*(W)$, called $I(W)(\cZ_*)$ \cite{Adamo:2011cb, Chicherinetal2016}. This integral is the equivalent of a Feynman integral of the interaction indicated by $W$.

Given a choice of $\cZ_*$, each integral $I(W)(\cZ_*)$ defines a volume on the space parametrized by the product of matrices $C_*(W) \cdot \cZ_*$. One of the goals of the Wilson loop approach to SYM N=4 theory is to write the sum of all $I(W)$, for a given number of propagators and external particles, as a volume of some geometric space. This is to parallel the story of the Ampltihedron \cite{Arkani-Hamed:2013jha}. However, in this section, we show that the union of the spaces parametrized by the $C_*(W)$ as a subspace of $\Gr(|\cP|, n+1)$ is often not an orientable space, and thus cannot have a well defined volume. This is an alternate approach to this problem is provided by \cite{HeslopStewart}, who does this explicitly by working directly with the $I(W)(\cZ_*)$.

\subsection{Fibers of the ``Delete a Column" Map \label{section:fibers}}
In order to understand the relation between $C_*(W)$ and $C(W)$, we begin by understanding what happens to a point in $\Gr(k,n)$ when, when representated as the row span of an $(n+1) \times k$ matrix, the last column is removed. This section proves some general facts about this map. Section \ref{bundle section} applies these findings to understand the geometry of $C_*(W)$.

Choosing an ordered basis, $b_1, \dots, b_{n+1}$ of $\rr^{n+1}$, points in $\Gr(k,n+1)$ may be represented by $(n+1) \times k$ matrices. Let $\pi:\rr^{n+1} \to \rr^{n}$ be the map which projects out the $n+1^{rst}$ coordinate. This extends to a map of Grassmannians,
\ba
\begin{array}{rcl}
\pi: \Gr(k,n+1) & \rightarrow & \Gr(k,n) \cup \Gr(k-1,n) \\
 V & \mapsto  & \pi(V) \end{array}.
\ea \label{projection map}
If $V$ is in a Schubert cell $\sigma_I \subset \Gr(k,n+1)$, where $n+1 \not \in I$, then $\pi(V) \in \Gr(k,n)$. If $n+1 \in I$, then $\pi(V) \in \Gr(k-1,n)$.
This section describes the fibers of $\pi$ over various subsets of $\Gr(k,n)$.

\begin{prop} \label{prop:bundle}
Write $I \in \binom{[n+1]}{k}$, and $\sigma_I$ the Schubert cell defined by this set. Let $\Omega_{n+1} \subset \Gr(k,n+1)$ be the subset of $\Gr(k, n+1)$ defined by \bas \Omega_{n+1} = \bigcup_{I \in \binom{[n+1]}{k}, \; n+1 \not \in I} \sigma_I \;. \eas Then $\Omega_{n+1}$ is a $k$-dimensional vector bundle over $Gr(k,n)$.
\end{prop}

\begin{proof}
Let $\pi$ be the projection map from $\R^{n+1} \rightarrow \R^n$ as above.

By construction, there are no points in $\Omega_{n+1}$ such that removing the last column in any matrix representation of it will drop the rank of the matrix. Therefore, $\pi$ is well defined on $\Omega_{n+1}$, and surjective onto the image, $\Gr(k, n)$.

To see that $\pi$ is continuous, recall that the usual topology on $\Gr(k,n)$ is the quotient topology on $M_{k \times n}^{\rk k}$, the set of real $k \times n$ matrices of rank $k$, under the the quotient map, $q$, defined by the $\textrm{GL}(k)$ action. Here $M_{k \times n}^{\rk k}$ is endowed with the subspace topology from $M_{k \times n}$. Therefore, and neighborhood, $N_x$ of a point $x \in \Gr(k,n)$ is an open set in $M_{k\times n}$. Since $\pi$ induces a smooth projection $M_{k \times (n+1)} \rightarrow M_{k \times n}$, $q \circ \pi^{-1} \circ q^{-1}(N_x)$ is open in $\Omega_{n+1}$.

The fiber of $\pi$ over any point is isomorphic to $\R^k$. To see that $\pi$ is locally trivializable, consider the standard manifold structure on $\Gr(k,n)$. For $J \in {[n] \choose k}$, define \bas U_J = \{x \in \Gr(k,n) | \Delta_J(x) \neq 0\}\;. \eas That is, $U_J$ is the set of all points in $\Gr(k,n)$ such that the $J^{th}$ Pl\"{u}cker coordinate is nonzero. These open sets $U_J$ form an atlas on $\Gr(k,n)$ when paired with the map $\phi_J$. Given any matrix representation of $x$, $M_x$, let $\Delta_J(M_x)$ be the $k \times k$ minor of $M_x$ defined by $J$. Then  $\phi_J(x)$ is defined as the image of $(\Delta_J(M_x))^{-1}M_x$ in $\R^{k(n-k)}$ under the map that removes the identity matrix in the columns indicated by $J$. One can create a similar atlas for $\Gr(k, n+1)$ to $\R^{k(n+1-k)}$. In fact, for the sets $J$ common to both $\Gr(k, n+1)$ and $\Gr(k, n)$, i.e. precisely those that do not contain $n+1$, this shows that the open set $U_J$ in $\Gr(k, n+1)$ is homeomorphic to a trivial real $k$ vector bundle over the $U_J$ in $\Gr(k, n)$.
\end{proof}

Next, we examine the structure of $\Omega_{n+1}$.

\begin{thm} \label{thm:deodhar_components_of_fiber}
Let $\cD$ be the Deodhar component in $\Gr(k,n)$ labelled by the Go-diagram $D$. Then, $\pi^{-1}(\cD)$ is the union of all Deodhar components $\cD'$ labelled by Go-diagrams $D'$ obtained by adding a column of boxes on the left of $D$.
\end{thm}

\begin{proof}
Let $p \in \pi^{-1}(\cD)$. Then, $p$ lies in some Deodhar component in $Gr(k,n+1)$. Theorem \ref{thm:go-networks} says there is some unique weighted Go-network $D'(p)$ representing $p$. Corollary \ref{cor:restricting_to_columns} implies that the projection $\pi(p)$ is obtained by deleting all the vertices in the leftmost column of $D'(p)$ and all edges incident to these vertices. Since $\pi(p) \in \cD$, this network obtained by deleting the left column of vertices must be $N(D)$.

Conversely, let $p \in \cD'$. The Deodhar component associated to some Go-diagram $D'$ obtained by adding a column of boxes to the left of $D$. Then, Theorem \ref{thm:go-networks} says that $p$ has a unique realization as a weighted Go-network $D'(p)$. Corollary \ref{cor:restricting_to_columns} says that the entries in the first $n$ rows of a matrix representing $p$ depend only on the part of the part of the Go-network $D'(p)$ agreeing with $D$. So, $\pi(p) \in D$.
\end{proof}

Next we study the Deodhar decomposition of $\Omega_{n+1}$, and the boundary structure therein.

\begin{prop} \label{prop:unique top dim}
Let $\cD$ be a Deodhar component in $Gr(k,n)$. The fiber $\pi^{-1}(\cD)$ contains a unique top dimensional Deodhar component.
\end{prop}

\begin{proof}
Let $\cD$ be the Deodhar component in $Gr(k,n)$ labelled by the Go-diagram $D$. Let $\cD' \subset \pi^{-1}(\cD)$ be the Deodhar component labelled by the Go-diagram $D'$. Theorem \ref{thm:deodhar_components_of_fiber} says that $D'$ is obtained by adding a column of $k$ boxes to on the left of $D$. So,
\begin{displaymath}
\mathrm{dim}\left(\cD'\right) \leq \mathrm{dim}\left(\cD \right) +k \;.
\end{displaymath}
\noindent
Corollary \ref{cor:dim_of_deodhar_component} implies that equality occurs if and only if the new column contains no $\circ$'s. There is a unique filling of the new column which does not use any $\circ$'s, which is produced by the following procedure:
%

\begin{itemize}
\item Let $b = (i, n+1)$ be a box in the new column of $D$ such that for all $j <i$, the box $(j, n+1)$ is filled with a $\bullet$ or $+$. 
\item Fill $b$ with a $+$.
\item If this violates the distinguished property, change this to a $\bullet$.
\item Repeat with $b = (i+1, n)$. \qedhere
\end{itemize}
\end{proof}

In Remark 7.11 in \cite{talaska:network_parametrizations}, Talaska and Williams give the following algorithm for constructing a weighted Go-network from a point in a Grassmannian. Given $V \in Gr(k,n)$, one may use this algorithm to determine which Deodhar component any point $V'$ in the fiber $\pi^{-1}(V)$ lies in.
\begin{itemize}
\item For each box $b$ in the new column starting from the bottom, if $\ell(u_{b^{in}}s_b ) < \ell(u_{b^{in}})$ fill $b$ with a $\bullet$ then proceed to the next box up.
\item If $\ell( u_{b^{in}}s_b ) > \ell(u_{b^{in}})$, compute $\Delta_{I_b}(V)$.
\item If $\Delta_{I_b}(V) = 0$, fill $V$ with a $\circ$. Otherwise, fill $b$ with a $+$. Proceed to the next box up.
\end{itemize}
 We next describe the boundary structure of the Deodhar components in the fiber $\pi^{-1}\left( \cD \right)$ over the Deodhar component $\cD \subset \Gr(k,n)$. $\cD'$ is on the boundary of $\cD$ if and only if all Pl\"ucker coordinates vanishing on $\cD$ also vanish on $\cD'$. We will need the following technical lemma.

\begin{lem} \label{lem:white_stones_do_not_change1}
Let $D$ be a Go-diagram. Suppose the box $b =(i, j)$ contains a $\circ$ or a $+$ and consider a box $c = (k, j)$ with $k > i$. That is, $c$ is below $b$ in the same column.
\begin{itemize}
\item[(i)] If the box $c$ contains a $+$ and $D'$ is the diagram obtained by replacing the $+$ in $c$ with a $\circ$, then the box $b$ does not violate the distinguished property in $D'$.
\item[(ii)] If the box $c$ contains a $\bullet$ and $D'$ is the diagram obtained by replacing the black stone in $c$ with a $+$, then the box $b$ does not violate the distinguished property in $D'$.
\end{itemize}
\end{lem}

\begin{proof}
We prove point (i) in the case where $b$ contains a $\circ$. The proofs for all other cases are similar and are left to the reader. Suppose that $b$ contains a $\circ$. Let $s_b = (i,i+1)$ and $s_c = (j,j+1)$ be the transpositions associated to boxes $b$ and $c$. Let
\bas
u_{b^{in}}^{D}(i) = x, \\
u_{b^{in}}^{D}(i+1) = y, \\
u_{c^{in}}^{D}(j) = z, \mbox{ and} \\
u_{c^{in}}^{D}(j+1) = w.
\eas
Since $b$ is filled with a $\circ$ and $c$ is filled with a $+$, $y > x$ and $w > z$. If $y \neq z$, $u_{b^{in}}^{D'}(i) = x$ and $u_{b^{in}}^{D'}(i+1) = y$. Then, $\ell(u_{b^{in}}^{D'} s_b) > \ell(u_{b^{in}}^{D'})$ and the white stone in $b$ does not violate the distinguished property. If $y = z$, $u_{b^{in}}^{D'}(i) = x$ and $u_{b^{in}}^{D'}(i+1) = w$. Since $w > z = y > x$, $\ell(u_{b^{in}}^{D'} s_b) > \ell(u_{b^{in}}^{D'})$ and the white stone in $b$ does not violate the distinguished property.
\end{proof}

\begin{thm} \label{thm:deodhar boundary struct}
Let $D'$ and $D''$ be Go-diagrams indexing Deodhar components $\cD'$ and $\cD''$ in the fiber $\pi^{-1}(\cD)$. Then, $\cD''$ is a codimension one boundary of $\cD'$ if and only if $D''$ is obtained by changing a single $+$ in a box $b$ in the leftmost column of $D''$ to a $\circ$, then reading up from $b$ in the leftmost column and changing $\bullet$'s to $+$'s as is necessary to avoid a violation of the distinguished property.
\end{thm}

\begin{proof}
Suppose that $D''$ is obtained from $D'$ in the manner described. Lemma \ref{lem:white_stones_do_not_change1} guarantees that $D''$ is in fact a Go-diagram. Theorem \ref{thm:sets_from_networks} implies that the set $I_c$ associated to the box $c = (i, n+1)$ in the leftmost column is identical in $D''$ and $D'$. Since the set of boxes containing $\circ$'s in the leftmost column of $D'$ is a subset of the set of boxes containing white stones in $D''$, $\cD''$ is on the boundary of $\cD'$. Counting the number of $+$'s and $\bullet$'s in $D'$ and $D''$, $\mathrm{dim}\left( \cD'' \right) = \mathrm{dim}\left( \cD' \right) - 1$.

Now, suppose that $\cD'$ and $\cD''$ are Deodhar components in the fiber $\pi^{-1}\left(\cD\right)$ labelled by the Go-diagrams $D'$ and $D''$, and that $\cD''$ is a codimension one boundary of $\cD'$. Since the sets $I_c$ labelling boxes in the leftmost columns of $D'$ and $D''$ are identical, the set of boxes containing $\circ$'s in the leftmost column of $D'$ must be a subset of the set of boxes containing $\circ$'s in $D''$. Since $\cD''$ has codimension one, exactly one square in the leftmost column of $D'$ must change to a $\circ$ in $D''$. Since in a Go-diagram a square is filled with a $\bullet$ if and only if it violates the distinguished property, $\bullet$'s cannot change to $\circ$'s. Then, one $+$ in a box $b$ must change to a $\circ$. Lemma \ref{lem:white_stones_do_not_change1} says that changing this plus to a white stone will not cause any of the $\circ$'s or $+$'s in the diagram to violate the distinguished property. However, it is possible that a $\bullet$ will no longer reduce the length of the associated subword after changing $b$ to a $\circ$. Since $\cD''$ is a codimension one boundary, any $\bullet$ which no longer violates the distinguished property must be changed into $+$.
\end{proof}

\begin{eg}
Let $\cD$ be the Deodhar component labelled by the following Go-diagram.
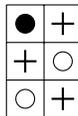
\begin{displaymath}
\begin{tikzpicture}
\begin{scope}[scale=.5]
\draw[step = 1] (0,0) grid (2,3);
%
\draw[fill = white] (.5,.5) circle (.25);
\draw[thick] (1.5,.2) -- (1.5,.8);
\draw[thick] (1.2,.5) -- (1.8,.5);
%
\draw[thick] (.5,1.2) -- (.5,1.8);
\draw[thick] (.2,1.5) -- (.8,1.5);
\draw[fill = white] (1.5,1.5) circle (.25);
%
\draw[fill = black] (.5,2.5) circle (.25);
\draw[thick] (1.5,2.2) -- (1.5,2.8);
\draw[thick] (1.2,2.5) -- (1.8,2.5);
\end{scope}
\end{tikzpicture}
\end{displaymath}
\noindent
The following is the boundary poset of Go-diagrams labelling Deodhar components in the fiber $\pi^{-1}\left(\cD\right)$.
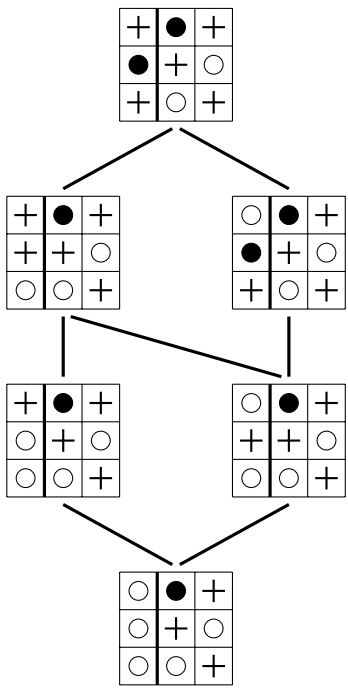
\begin{displaymath}
\begin{tikzpicture}
\begin{scope}[scale=.5]
\draw[step=1] (0,0) grid (3,3);
\draw[very thick] (1,0) -- (1,3);
%
\draw[thick] (.5,2.2) -- (.5,2.8);
\draw[thick] (.2,2.5) -- (.8,2.5);
\draw[fill = black] (.5,1.5) circle (.25);
\draw[thick] (.5,.2) -- (.5,.8);
\draw[thick] (.2,.5) -- (.8,.5);
%
\draw[fill = black] (1.5,2.5) circle (.25);
\draw[thick] (1.5,1.2) -- (1.5,1.8);
\draw[thick] (1.2,1.5) -- (1.8,1.5);
\draw[fill = white] (1.5,.5) circle (.25);
%
\draw[thick] (2.5,2.2) -- (2.5,2.8);
\draw[thick] (2.2,2.5) -- (2.8,2.5);
\draw[fill = white] (2.5,1.5) circle (.25);
\draw[thick] (2.5,.2) -- (2.5,.8);
\draw[thick] (2.2,.5) -- (2.8,.5);
\begin{scope}[xshift=-3cm, yshift=-5cm]
\draw[step=1] (0,0) grid (3,3);
\draw[very thick] (1,0) -- (1,3);
%
\draw[thick] (.5,2.2) -- (.5,2.8);
\draw[thick] (.2,2.5) -- (.8,2.5);
\draw[thick] (.5,1.2) -- (.5,1.8);
\draw[thick] (.2,1.5) -- (.8,1.5);
\draw[fill = white] (.5,.5) circle (.25);
%
\draw[fill = black] (1.5,2.5) circle (.25);
\draw[thick] (1.5,1.2) -- (1.5,1.8);
\draw[thick] (1.2,1.5) -- (1.8,1.5);
\draw[fill = white] (1.5,.5) circle (.25);
%
\draw[thick] (2.5,2.2) -- (2.5,2.8);
\draw[thick] (2.2,2.5) -- (2.8,2.5);
\draw[fill = white] (2.5,1.5) circle (.25);
\draw[thick] (2.5,.2) -- (2.5,.8);
\draw[thick] (2.2,.5) -- (2.8,.5);
\end{scope}
\begin{scope}[xshift=3cm, yshift=-5cm]
\draw[step=1] (0,0) grid (3,3);
\draw[very thick] (1,0) -- (1,3);
%
\draw[fill = white] (.5,2.5) circle (.25);
\draw[fill = black] (.5,1.5) circle (.25);
\draw[thick] (.5,.2) -- (.5,.8);
\draw[thick] (.2,.5) -- (.8,.5);
%
\draw[fill = black] (1.5,2.5) circle (.25);
\draw[thick] (1.5,1.2) -- (1.5,1.8);
\draw[thick] (1.2,1.5) -- (1.8,1.5);
\draw[fill = white] (1.5,.5) circle (.25);
%
\draw[thick] (2.5,2.2) -- (2.5,2.8);
\draw[thick] (2.2,2.5) -- (2.8,2.5);
\draw[fill = white] (2.5,1.5) circle (.25);
\draw[thick] (2.5,.2) -- (2.5,.8);
\draw[thick] (2.2,.5) -- (2.8,.5);
\end{scope}
\begin{scope}[xshift=-3cm, yshift=-10cm]
\draw[step=1] (0,0) grid (3,3);
\draw[very thick] (1,0) -- (1,3);
%
\draw[thick] (.5,2.2) -- (.5,2.8);
\draw[thick] (.2,2.5) -- (.8,2.5);
\draw[fill = white] (.5,1.5) circle (.25);
\draw[fill = white] (.5,.5) circle (.25);
%
\draw[fill = black] (1.5,2.5) circle (.25);
\draw[thick] (1.5,1.2) -- (1.5,1.8);
\draw[thick] (1.2,1.5) -- (1.8,1.5);
\draw[fill = white] (1.5,.5) circle (.25);
%
\draw[thick] (2.5,2.2) -- (2.5,2.8);
\draw[thick] (2.2,2.5) -- (2.8,2.5);
\draw[fill = white] (2.5,1.5) circle (.25);
\draw[thick] (2.5,.2) -- (2.5,.8);
\draw[thick] (2.2,.5) -- (2.8,.5);
\end{scope}
\begin{scope}[xshift=3cm, yshift=-10cm]
\draw[step=1] (0,0) grid (3,3);
\draw[very thick] (1,0) -- (1,3);
%
\draw[fill = white] (.5,2.5) circle (.25);
\draw[thick] (.5,1.2) -- (.5,1.8);
\draw[thick] (.2,1.5) -- (.8,1.5);
\draw[fill = white] (.5,.5) circle (.25);
%
\draw[fill = black] (1.5,2.5) circle (.25);
\draw[thick] (1.5,1.2) -- (1.5,1.8);
\draw[thick] (1.2,1.5) -- (1.8,1.5);
\draw[fill = white] (1.5,.5) circle (.25);
%
\draw[thick] (2.5,2.2) -- (2.5,2.8);
\draw[thick] (2.2,2.5) -- (2.8,2.5);
\draw[fill = white] (2.5,1.5) circle (.25);
\draw[thick] (2.5,.2) -- (2.5,.8);
\draw[thick] (2.2,.5) -- (2.8,.5);
\end{scope}
\begin{scope}[yshift = -15cm]
\draw[step=1] (0,0) grid (3,3);
\draw[very thick] (1,0) -- (1,3);
%
\draw[fill = white] (.5,2.5) circle (.25);
\draw[fill = white] (.5,1.5) circle (.25);
\draw[fill = white] (.5,.5) circle (.25);
%
\draw[fill = black] (1.5,2.5) circle (.25);
\draw[thick] (1.5,1.2) -- (1.5,1.8);
\draw[thick] (1.2,1.5) -- (1.8,1.5);
\draw[fill = white] (1.5,.5) circle (.25);
%
\draw[thick] (2.5,2.2) -- (2.5,2.8);
\draw[thick] (2.2,2.5) -- (2.8,2.5);
\draw[fill = white] (2.5,1.5) circle (.25);
\draw[thick] (2.5,.2) -- (2.5,.8);
\draw[thick] (2.2,.5) -- (2.8,.5);
\end{scope}
\draw[very thick] (1.4,-.2) -- (-1.5,-1.8);
\draw[very thick] (1.6,-.2) -- (4.5,-1.8);
\draw[very thick] (-1.5,-5.2) -- (-1.5,-6.8);
\draw[very thick] (4.5,-5.2) -- (4.5,-6.8);
\draw[very thick] (-1.3,-5.2) -- (4.3,-6.8);
\draw[very thick] (-1.5,-10.2) -- (1.4,-11.8);
\draw[very thick] (4.5,-10.2) -- (1.6,-11.8);
\end{scope}
\end{tikzpicture}
\end{displaymath}
\end{eg}

While we cannot restrict $\Omega_{n+1} \rightarrow \Gr(k,n)$ to a bundle that can be embedded into $\Grnn(k, n+1)$. However, we may  define a restriction of the map $\pi$ to a map of non-negative Grassmannians. \bas
\pi_{\geq 0} : Gr_{\geq 0}(k,n+1) \to Gr_{\geq 0}(k,n) \cup Gr_{\geq 0}(k-1,n).
\eas Deleting a column does not effect the positivity of any minors which do not involve the deleted column. However, in this restriction, there is no longer a bundle structure since the fibers of $\pi_{\geq 0}$ over different points are no longer equidimensional.

\begin{eg}
Consider the point
\bas
V =
\mathrm{span} \left(
\begin{array}{ccc}
1 & 0 & -1 \\
0 & 1 & 0
\end{array}
\right) \in \Grnn(2,3).
\eas
The fiber of $\pi$ over $V$ is
\bas
\pi^{-1}(V) =
\left\{
\mathrm{span} \left(
\begin{array}{ccc|c}
1 & 0 & -1 & g_1 \\
0 & 1 & 0 & g_2
\end{array}
\right) : (g_1, g_2) \in \rr^2 \right\}.
\eas A point in this fiber is positive if and only if $g_1 \leq 0$ and $g_2 = 0$. So, $\pi_{\geq 0}^{-1}(V)$ is only 1-dimensional. On the other hand, the point \bas W = \mathrm{span} \left(
\begin{array}{ccc}
1 & 1 & 0 \\
0 & 0 & 1
\end{array}
\right) \in \Grnn(2,3) \eas has a two dimensional fiber \bas
\pi_{\geq 0}^{-1}(W) =
\left\{
\mathrm{span} \left(
\begin{array}{ccc|c}
1 & 1 & 0 & g_1 \\
0 & 0 & 1 & g_2
\end{array}
\right) : g_1 < \R; \; g_2 > \R_2  \right\}.
\eas
\end{eg}

Let $\Sigma\subset \Grnn(k,n)$ be a positroid cell. We may understand the set $\pi_{\geq 0}^{-1}(\Sigma)$, in terms of Deodhar components as follows. Equation (\ref{eqn:positroid cell plucker}) says that each positroid cell is a semialgebraic subset of $\Grnn(k,n)$, defined by setting certain Pl\"ucker coordinates to zero and demanding other Pl\"ucker coordinates are uniformly non-negative. The fiber $\pi^{-1}(\cD)$ is the semialgebraic subset of $Gr(k,n+1)$ defined by the exact same equations that define $\cD$, only now these equations are viewed as equations in the Pl\"ucker coordinates on $Gr(k,n+1)$. No constraints are imposed on the Pl\"ucker coordinates $\Delta_I$ when $n+1 \in I$. For $\pi^{-1}_{\geq 0}(\Sigma)$, we must additionally impose $\Delta_I \geq 0$ when $n+1 \in I$. The structure of $\pi^{-1}_{\geq 0}(\Sigma)$ is described in the following lemma.

\begin{thm} \label{thm:fiber_boundary_structure_positroid}
Let $\Sigma \subset Gr_{\geq 0}(k,n)$ be the positroid cell given by the \Le-diagram $D$. Then,
\begin{itemize}
\item[(i)] $\pi^{-1}_{\geq 0}\left( \Sigma \right)$ contains a unique top dimensional positroid cell $\Sigma'$.
\item[(ii)] $\mathrm{dim}(\Sigma') = \mathrm{dim}(\Sigma) + k$ if and only if $D$ contains no $+$'s with $\circ$'s below them in the same column.
\item[(iii)] The boundary structure of all the positroid cells in $\pi^{-1}_{\geq 0}(\Sigma)$ is that of a Boolean lattice.
\end{itemize}
\end{thm}

\begin{proof}
Let $\Sigma \subset Gr_{\geq 0}(k,n)$ be the positroid cell given by the \Le-diagram $D$. Let $\cD$ be the Deodhar component labelled by $D$, so $\Sigma = \cD \cap Gr_{\geq 0}(k,n)$. Then,
\begin{displaymath}
\pi^{-1}_{\geq 0}\left( \Sigma \right) = \pi^{-1}\left( \cD \right) \cap Gr_{\geq 0}(k,n+1).
\end{displaymath}
\noindent
Let $\cD' \subset \pi^{-1}\left( \cD \right)$ be the Deodhar component labelled by the Go-diagram $D'$. So, $D'$ is obtained by adding a column of boxes to $D$. Theorem \ref{thm:positive_deodhar_components} says that $\cD' \cap Gr_{\geq 0}(k,n+1)$ is nonempty if and only if $D'$ is a Le-diagram. Let $b$ be a box in the new column added to create $D'$. If there is any box containing a $\circ$ to the right of $b$ in its row with a $+$ above it, then filling $b$ with a plus will cause a violation of the Le-property. So, such boxes must be filled with $\circ$'s, and all other boxes in the new column may be safely filled with $\circ$'s or $+$'s. A unique top dimensional cell is then obtained by filling all boxes which may be filled with $+$'s in the new column with $+$s, proving points (i) and (ii).

Now, let $D'$ and $D''$ be two Le-diagrams labelling positroids $\Sigma'$ and $\Sigma''$ in the fiber $\pi^{-1}_{\geq 0}\left( \Sigma \right)$. We claim that $\Sigma''$ is a codimension $1$ boundary of $\Sigma'$ if and only if $\Sigma''$ is obtained by changing one of the $+$'s in the leftmost column of $\Sigma'$ into a $\circ$. The cell $\Sigma''$ is on the boundary of $\Sigma'$ if and only if all Pl\"ucker coordinates vanishing on $\Sigma'$ also vanish on $\Sigma''$. From Theorem \ref{thm:sets_from_networks}, we see that the set $I_b$ labelling a box $b$ in the new column does not depend on the filling of the boxes in the new column. So, $\Sigma''$ is on the boundary of $\Sigma'$ if and only if all boxes containing $\circ$'s in the new column in $D'$ also contain $\circ$'s in $D''$. The dimension of a positroid cell is the number of $+$'s in its Le-diagram, proving the claim. Then, the boundary poset of the cells in the fiber $\pi^{-1}_{\geq 0}(\Sigma)$ is exactly the Boolean lattice of subsets of boxes containing $+$'s in the new column of the Le-diagram labelling the top dimensional cell of $\pi^{-1}_{\geq 0}(\Sigma)$.
\end{proof}

\subsection{The bundle structure of $C_*(W)$ \label{bundle section}}

Let $W = (\cP, n)$ be a Wilson loop diagram, with $|\cP| = k$.

In \cite{wilsonloop}, the authors show that the Wilson loop diagrams define matroids (realized by the matrices $C(W)$). Furthermore, if a Wilson loop diagram is admissible, then this matroid is a positroid. Let $\Sigma(W)$ be the subspace of $\Grnn(k, n)$ parametrized by the Wilson loop. In \cite{generalcombinatorics}, the authors provide an algorithm to read a Grassmann necklace from a Wilson loop diagram, and \cite{GNtoLe} gives an algorithm for reading off a Le-diagram from a Grassmann Necklace. In \cite{generalcombinatorics}, the authors show that $\Sigma(W)$ is a $3k$ dimensional space.

Define \ba \cW_{k,n} = \bigcup_{\begin{subarray}{c} W =(\cP, n) \; admis. \\ |\cP| =k \end{subarray}} \overline{\Sigma(W)}\;. \label{Wkn eq}\ea  In the small case of $\Grnn(2,6)$, \cite{casestudy} describes the geometry and topology of $\cW_{2,6}$. In this section we consider the subspace $\pi^{-1}\left(\cW_{k,n}\right)$ of $\Gr(k,n+1)$ collectively parametrized by $C_*(W)$.

\begin{lem}Given an admissible Wilson loop diagram and the map $\pi$ as in \eqref{projection map} the matrix $C_*(W)$ parametrizes $\pi^{-1}(\Sigma(W))$. \label{lem C star param}\end{lem}

\begin{proof}
By Definition \ref{dfn C}, if $x \in \Sigma(W)$, then $x$ can be represented by a matrix $M_x$ formed by evaluating the entries of $C(W)$ at appropriate points so that its Pl\"ucker coordinates are all non-negative.
Then, $\pi^{-1}(x)$ consists of all points in $\Gr(k, n+1)$ that can be expressed as $M_x$ with an extra column appended. These points are all in the space parametrized by $C_*(W)$. \end{proof}

Using the results of the previous section, we may break $\pi^{-1}(\Sigma(W))$ into a disjoint union of Deodhar components.

\begin{thm} \label{thm:4k dim diagrams}
The space $\pi^{-1}(\Sigma(W))$ is a $4k$-dimensional subspace of $\Gr(k,n)$.
\end{thm}
\begin{proof}
From Corollary \ref{cor:dim_of_deodhar_component}, the dimension of a Go-diagram is computed by counting its number of $+$'s and $\bullet$'s. In \cite{generalcombinatorics}, the authors show that $\Sigma(W)$ is a $3k$-dimensional subspace of $\Grnn(k,n)$. In Proposition \ref{prop:unique top dim}, we see that $\pi^{-1}(\Sigma(W))$ has an unique top dimensional cell, the  Go-diagram for which is formed by adding $k$ $+$'s and $\bullet$'s to the Go-diagram for $\Sigma(W)$. Therefore, the total dimension is $4k$.
\end{proof}

Furthermore, Theorem \ref{thm:deodhar_components_of_fiber} describes the Deodhar components that constitute $\pi^{-1}(\Sigma(W))$, and Theorem \ref{thm:deodhar boundary struct} describes how these components glue together.

However, when studying the geometry of the space spanned by all matrices of the form $C(W)$, one is interested in studying the space $\cW_{k,n}$, as defined in equation \eqref{Wkn eq}. Similarly, to understand the full geometry of the $C_*(W)$ matrices, one must consider $\pi^{-1}(\cW_{k,n})$.

From the definitions, one writes \bas \pi^{-1}(\cW_{k,n}) = \bigcup_{\begin{subarray}{c} W =(\cP, n) \; admis. \\ |\cP| =k \end{subarray}} \pi^{-1}(\overline{\Sigma(W)}) \; .\eas However, as little is know about the boundary cells of $\Sigma(W)$ without resorting to explicit calculation, this representation is not the most enlightening. It is useful therefore to note that the closure of the cells commutes with the pre-image of the projection map.

\begin{thm} \label{closure commutes}
One can write $\pi^{-1}(\overline{\Sigma(W)}) = \overline{\pi^{-1}(\Sigma(W))}$.
\end{thm}

\begin{proof}
First note that, for any positroid cell $\Sigma \subset \Grnn(k,n)$, \bas \pi^{-1}(\overline{\Sigma} ) = \bigcup_{\Sigma ' \subset \overline{\Sigma}} \pi^{-1}(\Sigma')\;.   \eas

For a positroid cell $\Sigma$, let $\cI_\Sigma$ be the associated Grassmann necklace. Then, $\Sigma$ is the subset of $\Gr(k,n)$ defined by
\begin{equation} \label{eqn:before_closing}
\Delta_J(x)  \begin{cases} > 0 &\textrm{ if } J \in \cI_\Sigma \\  = 0 &\textrm{ if } J \sgale{i} I_i \, \forall I_i \in \cI_\Sigma \\ \geq 0 &\textrm{ if } I_i \gale{i} J \, \forall I_i \in \cI_\Sigma. \end{cases} \;
\end{equation}
for all $x \in \Sigma$. Then, the closure $(\overline\Sigma)$ is defined by closing these inequalities as follows:
\begin{equation} \label{eqn:after_closing}
\Delta_J(x)  \begin{cases} \geq 0 &\textrm{ if } J \in \cI_\Sigma \\  = 0 &\textrm{ if } J \gale{i} I_i \, \forall I_i \in \cI_\Sigma \\ \geq 0 &\textrm{ if } I_i \gale{i} J \, \forall I_i \in \cI_\Sigma \end{cases}. \;
\end{equation}for all $x \in \Sigma$. 

Then, $\pi^{-1}(\overline \Sigma)$ is the subset of $\Gr(k,n+1)$ defined by the inequalities (\ref{eqn:after_closing}) with no new constraints introduced on Pl\"ucker coordinates $\Delta_I$ when $n+1 \in I$. Similarly, $\pi^{-1}(\Sigma)$ is the subset of $\Gr(k,n+1)$ defined by the inequalities (\ref{eqn:before_closing}). Taking the closure gives
\begin{displaymath} \overline{\pi^{-1}(\Sigma)} = \pi^{-1}(\overline \Sigma) \;. \qedhere \end{displaymath}
\end{proof}

As a corollary of Proposition \ref{prop:bundle}, the $C_*(W)$ parametrize a vector bundle over the subspace of $\Grnn(k,n)$ parametrized by the $C(W)$.

\begin{cor} \label{cor:WLD bundle}
The space $\pi^{-1}(\cW_{k,n})$ has the structure of a real $k$-vector bundle over $\cW_{k,n}$.
\end{cor}


Finally, we note that while $\Sigma(W)$ is a subset of the positive Grassmannian, $\pi^{-1}(\cW_{k,n})$ is not. In other words, the matrices $C_*(W)$ do not, in general, define positroids, illustrated the following example.

\begin{eg} \label{eg:C star not pos}
Consider the Wilson loop diagram \bas W = {\makediag{1}{0}{5}{0}{2}{0}{4}{0}}\;. \eas Then, the space parametrized by \bas C_*(W) = \begin{pmatrix}
                                                                                                           a & b & 0 & 0 & c & d & 1 \\
                                                                                                           0 & e & f & g & h & 0 & 1
                                                                                                         \end{pmatrix}\eas can never intersect $\Grnn(2,6)$. For instance, if $\Delta_{17} > 0$ holds,  the variable $a$ must be positive. Similarly, if $\Delta_{37} > 0$ holds, this forces $f$ to be negative. However, this creates a negative minor, $\Delta_{13} = af <0$.
\end{eg}

This phenomenon can be understood by Theorem \ref{thm:fiber_boundary_structure_positroid} which gives a condition for when the preimage of of a cell $\Sigma \subset \Grnn(k,n)$ under $\pi$ does not intersect the positive Grassmannian in its full dimension. Namely, we have the following statement:

\begin{thm} \label{thm:when C star not pos}
Let $W$ be an admissible Wilson loop diagram with $k$ propagators. If the Le-diagram associated to $W$, $D(W)$, has a column with a $+$ square above a square with a $\circ$, then \bas \dim( \pi^{-1}(\Sigma(W)) ) > \dim( \pi^{-1}_{\geq 0}(\Sigma(W)) ). \eas \end{thm}

\begin{proof}
If $D(W)$ has a $+$ square directly above a square with a white stone, by Theorem \ref{thm:fiber_boundary_structure_positroid}, the dimension of
$\pi_{\geq 0} ^{-1}(\Sigma(W)) < 4k$. However, by Theorem \ref{thm:4k dim diagrams}, $\pi^{-1}(\Sigma(W))$ is $4k$ dimensional.
\end{proof}

\begin{rmk}
From a combinatorial standpoint, Theorem \ref{thm:when C star not pos} can be interpreted as giving a condition for when the matroid represented by a generic point in $\pi^{-1}(\Sigma(W))$ is a positroid.
\end{rmk}

\begin{eg}
We return to Example \ref{eg:C star not pos} to illustrate the above theorem. For \bas W = {\makediag{1}{0}{5}{0}{2}{0}{4}{0}}\;, \eas we write \bas C(W) = \begin{pmatrix}
                                                                                                           a & b & 0 & 0 & c & d \\
                                                                                                           0 & e & f & g & h & 0
                                                                                                         \end{pmatrix}.\eas
The associated Le-diagram is
\begin{displaymath}
\begin{tikzpicture}
\begin{scope}[scale=.5]
\draw (-2,1) node {$D(W) =$};
\draw[step=1] (0,0) grid (4,2);
%
\draw[fill = white] (0.5,0.5) circle (.25);
\draw[thick] (1.5,.2) -- (1.5,.8);
\draw[thick] (1.2,.5) -- (1.8,.5);
\draw[thick] (2.5,.2) -- (2.5,.8);
\draw[thick] (2.2,.5) -- (2.8,.5);
\draw[thick] (3.5,.2) -- (3.5,.8);
\draw[thick] (3.2,.5) -- (3.8,.5);
%
\draw[thick] (.5,1.2) -- (.5,1.8);
\draw[thick] (.2,1.5) -- (.8,1.5);
\draw[fill = white] (1.5,1.5) circle (.25);
\draw[thick] (2.5,1.2) -- (2.5,1.8);
\draw[thick] (2.2,1.5) -- (2.8,1.5);
\draw[thick] (3.5,1.2) -- (3.5,1.8);
\draw[thick] (3.2,1.5) -- (3.8,1.5);
\draw (4.2,.2) node {$.$};
\end{scope}
\end{tikzpicture}
\end{displaymath}

For details on this calculation, see \cite{casestudy}. In column $6$, there is a $+$ square above a $\circ$ square. Furthermore, we see that the top dimensional Deodhar component of $\pi^{-1}(\Sigma(W))$, call it $D_*(W)$, has the Go diagram
\begin{displaymath}
\begin{tikzpicture}
\begin{scope}[scale=.5]
\draw (-3,1) node {$D_{*}(W) =$};
\draw[step=1] (-1,0) grid (4,2);
%
\draw[thick] (-.5,.2) -- (-.5,.8);
\draw[thick] (-.2,.5) -- (-.8,.5);
\draw[fill = black] (-.5,1.5) circle (.25);
%
\draw[fill = white] (0.5,0.5) circle (.25);
\draw[thick] (1.5,.2) -- (1.5,.8);
\draw[thick] (1.2,.5) -- (1.8,.5);
\draw[thick] (2.5,.2) -- (2.5,.8);
\draw[thick] (2.2,.5) -- (2.8,.5);
\draw[thick] (3.5,.2) -- (3.5,.8);
\draw[thick] (3.2,.5) -- (3.8,.5);
%
\draw[thick] (.5,1.2) -- (.5,1.8);
\draw[thick] (.2,1.5) -- (.8,1.5);
\draw[fill = white] (1.5,1.5) circle (.25);
\draw[thick] (2.5,1.2) -- (2.5,1.8);
\draw[thick] (2.2,1.5) -- (2.8,1.5);
\draw[thick] (3.5,1.2) -- (3.5,1.8);
\draw[thick] (3.2,1.5) -- (3.8,1.5);
\draw (4.2,.2) node {$.$};
\end{scope}
\end{tikzpicture}
\end{displaymath}
Since this diagram has a $\bullet$, the associated Deodhar component does not intersect $\Grnn(k,n+1)$. Furthermore, one may check explicitly that for all values of $a, b, c, d, e, f, g$ and $h$ such that $C(W) \in \Grnn(k,n)$, the matrix \bas C_*(W) = \begin{pmatrix}
                                                                                                           a & b & 0 & 0 & c & d & 1 \\
                                                                                                           0 & e & f & g & h & 0 & 1
                                                                                                         \end{pmatrix}\eas represents a point in the Deodhar component represented by $D_*(W)$. 
\end{eg}

\subsection{Boundaries \label{boundary chasing section}}
For the remainder of this paper, we show that the subspace of $\Gr(k,n+1)$ parametrized by the $C_*(W)$, is not orientable. To do this, we show that $\pi^{-1}(\cW_{k,n})$, is not an orientable bundle over $\cW_{k,n}$.

First however, we introduce an interesting family of Wilson loop diagrams introduced to us by Paul Heslop, \cite{Heslopcommunication}. In \cite{casestudy}, the authors provide a graphic device for determining which Wilson Loop diagrams share codimension $1$ boundaries in $\Grnn(2,6)$.

\begin{dfn}\label{def:boundaryprops}
Let $W = (\cP, n)$ be an admissible Wilson loop diagram, and $p \in \cP$ one of its propagators. For $v \in V_p$, the {\bf boundary propagator} $\D_v p$ is obtained by moving the endpoint of $p$ away from vertex $v$ while maintaining the requirement that no two propagators cross each other, and that a collection of $k>1$ propagators and boundary propagators are supported on at least $k+3$ vertices until one of the following occurs:
\begin{enumerate}
\item the endpoint of $p$ reaches another vertex, i.e. $\D_vp$ is supported on $V_p \backslash {v}$; or
\item the endpoint of $p$ touches the endpoint of another propagator $q$.
\end{enumerate}

Define the {\bf boundary diagram} $\D_{p,v}(W)$ to be the diagram obtained from $W$ by replacing propagator $p$ with $\D_vp$.
\end{dfn}

\begin{eg}\label{boundaryeg}
Consider the Wilson loop diagram,
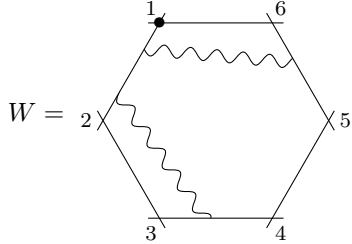
\[ W =\begin{tikzpicture}[rotate=60,baseline=(current bounding box.east)]
	\begin{scope}
	\drawWLD{6}{1.5}
	\drawnumbers
\drawprop{1}{-1}{5}{0}
\drawprop{1}{1}{3}{0}
		\end{scope}
	\end{tikzpicture}\]
and the propagator $p = (1,5)$. By replacing $p$ with $\D_2p$ and $\D_1p$ respectively, we obtain examples of both types of boundary diagrams:
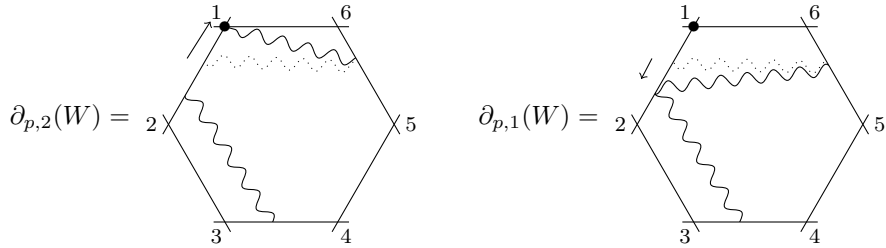
\begin{gather*}\D_{p,2}(W)  =
\begin{tikzpicture}[rotate=60,baseline=(current bounding box.east)]
	\begin{scope}
	\drawWLD{6}{1.5}
	\drawnumbers
	\modifiedprop{1}{-0.5}{5}{0}{propagator,dotted}
	\boundaryprop{5}{0}{1}{propagator}
	\drawprop{1}{0.5}{3}{0}
	\boundA{1}{-0.5}{1}
		\end{scope}
	\end{tikzpicture}
\qquad
\D_{p,1}(W) =
\begin{tikzpicture}[rotate=60,baseline=(current bounding box.east)]
	\begin{scope}
	\drawWLD{6}{1.5}
	\drawnumbers
	\drawprop{1}{0.3}{3}{0}
	\modifiedprop{5}{0}{1}{-0.3}{propagator,dotted}
	\modifiedprop{5}{0}{1}{1.2}{propagator}
	\boundB{1}{-0.5}{1}
		\end{scope}
	\end{tikzpicture}
\end{gather*} On the other hand, the in the following diagram,
for $p = (1,5)$ the boundary $\D_{p,5}(W')$ is not permitted as its two propagators are supported on only $4$ vertices.
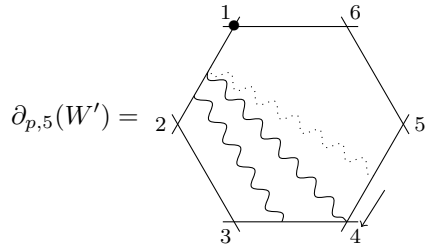
\begin{gather*}\D_{p,5}(W')  =
\begin{tikzpicture}[rotate=60,baseline=(current bounding box.east)]
	\begin{scope}
	\drawWLD{6}{1.5}
	\drawnumbers
	\modifiedprop{1}{-0.5}{4}{0}{propagator,dotted}
	\boundaryprop{1}{-0.5}{4}{propagator}
	\drawprop{1}{0.5}{3}{0}
	\boundA{4}{-0.5}{4}
		\end{scope}
	\end{tikzpicture}
\end{gather*}
\end{eg}

Since the support of a propagator determines which entries of $C(W)$ are nonzero, we define the matrix associated to these propagator moves as setting minors of $C(W)$ to $0$.  This is either a $1 \times 1$ or a $2 \times 2$ minor, written $\Delta_{p,v}(W)$:
\[\Delta_{p,v}(W) =\left\{ \begin{array}{ll}
c_{p,v} & \text{ if $p$ is no longer supported on $v$ in $\D_{p,v}(W)$}; \\
c_{p,v}c_{q,v+1} - c_{q,v}c_{p,v+1} & \text{ if propagators $p$ and $q$ touch on edge $v$ $\D_{p,v}(W)$.}
\end{array}\right.\]
Using this notation, we can write \bas C(\D_{p,v}(W)) = \lim_{\Delta_{p,v} \rightarrow 0} C(W) \;. \eas We call $C(\D_{p,v}(W))$ a {\bf boundary matrix} of $C(W)$. We may now define when two diagrams share a codimension one boundary.

\begin{conj}\label{mnemonic}
Let $W = (\cP, n) $ and $W' = (\cP', n)$ be two Wilson loop diagrams. If there exist two vertex propagators pairs $(p,v)$ and $(p', v')$, with $p \in \cP$, $v \in V_p$ and $p' \in \cP'$, $v' \in V_{p'}$ such that \bas C(\D_{p,v} (W))= C(\D_{p',v'} (W'))\;,\eas then the corresponding cells $\Sigma(W)$ and $\Sigma(W')$ share a codimension 1 boundary in $\Grnn(k,n)$.
\end{conj}

This statement is shown explicitly in the case of two propagators and six vertices in \cite{casestudy}. However, this statement remains conjectured in the general case. It is worth noting that not every codimension $1$ boundary of $\Sigma(W)$ can be obtained via these boundary diagram. Nor do these boundary diagrams capture every instance of when two cells $\Sigma(W)$ and $\Sigma(W')$ share a codimension 1 boundary in $\Grnn(k,n)$. These caveats aside, however, the boundary diagrams provide a good visual aide in understanding the orientability of $\cW_{k,n}^{*,+}$.


We define a family of Wilson loop diagrams that share codimension one boundaries with each other by Conjecture \ref{mnemonic}.

\begin{dfn}
Let $W = (\cP, n)$ be an admissible Wilson loop diagram. Let $p = (i,j)\in \cP$. Define another propagator $p'$
\bas p' =
\begin{cases}
(i-1, j)  \textrm{ or } (i, j-1) & \textrm{if } ; \; j \not \in \{i+2, i-2\} \\
(i-1, j-1) & \textrm{if }  \; j \in \{i+2, i-2\}. \\
\end{cases} \;
\eas We say that $p'$ is formed by {\textbf{moving $p$ clockwise}}. Call $W'$ be the Wilson loop diagram (not necessarily admissible) diagram formed by moving $p$ clockwise: $W' = ((\cP \setminus p)\cup p', n)$.
\end{dfn}

We say that $W$ has a \emph{valid clockwise move} if there is some $p \in \cP$ such that moving $p$ clockwise to $p'$ results in an admissible $W'$.
In particular, if $W$ has a propagator with a valid clockwise move, let $v \in V_p$ be the vertex in the support of $p$ that is not in the support of $p'$. Similarly, let $v' \in V_{p'}$ be the vertex in the support of $p'$ that is not in the support of $p$. Then, by construction \bas C(\D_{p, v}(W)) = C(\D_{p', v'} (W')) \;.\eas That is, by Conjecture \ref{mnemonic}, the cells $\Sigma(W)$ and $\Sigma(W')$ share a codimension 1 boundary.

We use these valid clockwise moves to rotate two families of Wilson loop diagrams to an isomorphic diagram with the order of the propagators changed.
\begin{dfn} \label{dfn series parallel}
Define \bas \cP_{series} = \{(i_1, j_1) \ldots (i_k, j_k) | i_r <_{i_1} j_r \; ; j_r \leq_{i_1} i_{r+1}\}\eas to be a set of propagators. Define $W_{series} = (\cP_{series} , n > 2k)$ to be a Wilson loop diagram defined by a propagator set of this form.
Define \bas\cP_{parallel} = \{(i_1, j_1) \ldots (i_k, j_k) | i_r <_{i_1} j_r \; ; j_k \leq_{i_1} j_{k+1} \textrm{ and } i_{k+1} \leq_{i_1} i_k \}\eas to be a set of propagators. Define $W_{parallel} = (\cP_{parallel} , n)$ and $W_{series} = (\cP_{series} , n)$ to be Wilson loop diagrams defined by propagator sets of this form.
\end{dfn}

Note that in $W_{series}$ and $W_{parallel}$, we introduce an ordering to the propagators not present before. 

\begin{lem} \label{lem:series and parallel}
Let $W_s$ be a diagram of the form $W_{series}$ and $W_p$ a diagram of the form $W_{parallel}$. There is a series of valid clockwise moves that rotate the propagators of $W_p$ until  $W_p = (\cP', n)$, with $\cP'$ is the same set of propagators with the order inverted. Similarly, there is a series or valid clockwise moves that rotate the propagators of $W_s$ until $W_s = (\cP', n)$, with $\cP'$ is the same set of propagators shifted by one. That is, $p_r = (i_{r-1}, j_{r-1}) \in \cP'$ corresponds to $p_{r-1}$ in $\cP$.
\end{lem}

\begin{proof}
Given a propagator $p = (i,j)$, label its two ends, $(p, i)$ and $(p,j)$. In this manner, we may distinguish valid clockwise propagators moves made on the edge $(i,i+1)$ of the Wilson loop diagram from those made on the edge $(j,j+1)$.

First, order the ends of the propagators as follows \ba p_{1, i_1} \prec p_{1, j_1} \prec p_{2, i_2} \prec \ldots \prec p_{k, j_k} \;.\label{ordered ends}\ea

To see the set of moves necessary to transform $W_s$ to $W_s'$, start with the first end of a propagator in the order given in \eqref{ordered ends} that has a a valid clockwise move. Move this end as many steps as possible via valid clockwise moves. This is possible since $n > 2k$. Then proceed to the next propagator with a valid clockwise move in the ordering \eqref{ordered ends} move this as far as possible via valid clockwise moves, or until it reaches its destination ($p_{r, j_r} = p_{r+1, j_{r+1}}$ or $p_{r, i_r} = p_{r+1, i_{r+1}}$). Continue this algorithm (cycling though ordering \eqref{ordered ends} when one reaches the end) always choosing the first propagator end that has a valid clockwise move or that hasn't reached its destination.

For $W_p$, if either $p_{1, i_1}$ or $p_{1, j_1}$ is the first propagator end with a valid clockwise move, move it as far a possible, or until it reaches its destination ($p_{1, j_1} = p_{k, i_k}$ or $p_{1, i_1} = p_{k, j_{k}}$). If not, move the first propagator end with a valid clockwise move once. The existence of such a move is guaranteed by the fact that $n \geq k+4$. Repeat this process until both ends of $p_1$ are in position. Once $p_r$ has been moved to the position originally occupied by $p_{k-r+1}$ (when $p_{r, j_r} = p_{k-r+1, i_{k-r+1}}$ or $p_{r, i_r} = p_{k-r+1, j_{k-r+1}}$), then do not move the endpoints of $p_r$ further. If either $p_{r+1, i_{r+1}}$ or $p_{r+1, j_{r+1}}$ is the first propagator end with a valid clockwise move, move it as far a possible, or until it reaches its destination. Otherwise, move the first propagator end with a valid clockwise move that has not reached its destination once. Repeat this process until both ends of $p_{r+1}$ are in position.
\end{proof}

Note that this reordering of propagators within the diagram has no effect on $\Sigma(W)$. The existence of this process, does, however, have an effect on the space $\pi^{-1}(\Sigma(W)$, parametrized by $C_*(W)$ as we show below.

\subsection{The orientation of the bundle \label{orientation section}}

In this section, we show that $\pi^{-1}(\cW_{k,n})$ is not orientable for certain $k$ and $n$.

We define a set of open sets on $\cW_{k,n}$ that we use in the proof of non-orientability.
\begin{dfn}
If $W_1$ and $W_2$ are two admissible Wilson loop diagrams such that $\Sigma(W_1)$ and $\Sigma(W_2)$ share a codimension one boundary, $V$, define $U_{W_1, W_2} = \Sigma(W_1) \cup \Sigma(W_2) \cup V$. Let $J$ be the lexicographically minimal set in ${[n]} \choose k$ such that the $J^{th}$ Pl\"{u}cker coordinate of points in $U_{W_1, W_2}$ is nonzero. Let the homeomorphisms $\phi_J$ as defined in the proof of Proposition \ref{prop:bundle}.
\end{dfn}

If two Wilson loop diagrams share multiple distinct codimension one boundaries, then there are multiple such open sets.  Let $\cA$ be any atlas on $\cW_{k,n}$ that contains all $(U_{W_1, W_2}, \phi_J)$ such that $W_1$ and $W_2$ have $k$ propagators and $n$ vertices.

We are now ready to prove that $\pi^{-1}(\cW_{k,n})$ is not orientable for certain $k$ and $n$.

\begin{thm}\label{thm:non-orientable}
The bundle $\pi^{-1}(\cW_{k,n})$ is not orientable in the following two situations:
\begin{enumerate}
\item If $n >2k$ and $k$ is even.
\item If $\lfloor \frac{k}{2} \rfloor$ is odd.
\end{enumerate}
\end{thm}

\begin{proof} 
It is sufficient to show that there exists some pair of charts $(U_{W_1, W_2}, \phi_J)$ and $(U_{W'_1, W'_2}, \phi_{J'})$ such that the structure group on $U_{W_1, W_2} \cap U_{W'_1, W'_2}$ cannot be reduced to $\textrm{Gl}_+(n)$. By definition, the index set $J$ is contained in ${[n]} \choose k$. That is, $n+1$ is not in any $J$ defining a chart. Since Corollary \ref{cor:WLD bundle} says $\pi^{-1}(\cW_{k,n}) \rightarrow \cW_{k,n}$ is a sub-bundle of $\Omega_{n+1} \rightarrow \Gr(k,n)$, we see from Proposition \ref{prop:bundle} that $\pi^{-1}(\cW_{k,n})$ is trivializable over $\cW_{k,n}$ on each $U_{W'_1, W'_2}$. Then we write \bas \phi_J^{-1}\circ \phi_{J'} :  (U_{W_1, W_2} \cap U_{W'_1, W'_2}) \times \R^k & \rightarrow (U_{W_1, W_2} \cap U_{W'_1, W'_2}) \times \R^k  \\ (C(W), \vec{v}) & \mapsto (C(W), t_{J, J'} \vec{v}) \;\eas for some gauge matrix $t_{J, J'} \in \textrm{Gl}(k)$. It is sufficient to show that there exists some transition map for which $t_{J, J'} \not \in \textrm{Gl}_+(k)$.

Now, consider the two admissible Wilson loop diagrams, $W_p$ and $W_s$ of the form $W_{series}$ and $W_{parallell}$ respectively. Let $H(W_p)$ and $H(W_s)$ be the two series of diagrams the come out of the valid clockwise moves described in Lemma \ref{lem:series and parallel}. Write $H(W_p) = \{W_p = W_1 , \ldots W_r = W_p\}$ and similarly for $H(W_s)$. One gets from $W_i$ to $W_{i+1}$ via a valid clockwise move.

Since $\phi_J$ is defined on matrix representations of points in $\Gr(k,n)$, we work with matrix representations. For $x \in \pi^{-1}(\Sigma(W_p))$ we write $C_*(W_p(x))$ with the rows ordered according to the ordering prescribed in Definition \ref{dfn series parallel}. For $x \in \pi^{-1}(\Sigma(W_s))$ we write $C_*(W_s(x))$ with similarly ordered rows. Moving from a point in $\pi^{-1}\Sigma(W_l)$ to a point in $\pi^{-1}\Sigma(W_{l+1})$ the clockwise move of a propagators $p_m$ is represented by a change in the positioning of the non-zero entries in only the $m^{th}$ rows of $C_*(W_l)(x)$ and $C_*(W_{l+1})(x)$. Moreover, only the first $n$ entries of the $m^{th}$ row potentially change. The last column of $C_*(W_l)(x)$ is uneffected by moving the propagator.

By Conjecture \ref{mnemonic}, $W_l$ and $W_{l+1}$ share a boundary. Write the chart in this case as $(U_{W_l, W_{l+1}}, \phi_{J_l})$. In the trivialization over $U_{W_l, W_{l+1}}$, we may write $C_*(W_l)(x)$ as $(C(W_l)(x), v(x))$. For $l \neq r-1$, the $G$-structure on $U_{W_{l-1}, W_l} \cap U_{W_l, W_{l+1}}$ is given by $t_{l-1, l} = \Big(\Delta_{J_{l-1}}(C(W_l)(x))\Big)^{-1}\Big( \Delta_{J_{l}}(C(W_l)(x))\Big)$.
Then \bas \phi_{J_{l}}^{-1} \circ \phi_{J_{l-1}}\Big(C(W_l)(x), \vec{v}(x)\Big) =  \Big(C(W_l)(x), \Big(\Delta_{J_{l-1}}(C(W_l)(x))\Big)^{-1}\Big( \Delta_{J_{l}}((C(W_l)(x))\Big) \vec{v})\Big) \;.\eas Since $C(W_l) C(W_{l-1})\in \Grnn(k,n)$, $\Delta_{J_{l-1}}(C_*(W_l))$ and $\Delta_{J_{l}}(C_*(W_l))$ both have positive determinants.

It remains to consider when $l=r$. Consider $H(W_{s})$, which is only defined when $n > 2k$. When $l = r$, we are at the last step of Lemma \ref{lem:series and parallel} when $W_r = W_1$ with the propagators cyclically shifted by one. In terms of the matrices, this means that the rows of $C(W_1)$ and $C(W_r)$ are related by a permutation matrix, $\sigma \in S_k$ that cyclically shifts the rows of $C(W)$. In this case, the $G$-structure on $U_{W_{r-1}, W_r} \cap U_{W_1, W_{2}}$ is given by $\Delta_{J_{1}}(C(W_1)(x))^{-1} \sigma \Delta_{J_{r-1}}(C(W_r)(x))$. This has positive determinant if and only if $\det(\sigma) > 0$. In other words, if and only if $k$ is even. This proves the first point.

The second point comes from the family $H(W_{p})$. When $r= l$, the diagrams $W_r = W_1$ are equal, but with the propagators order inverted. Again, this implies that $C(W_1)$ and $C(W_r)$ are related by a permutation matrix, $\sigma \in S_k$ that inverts the order of the rows of $C(W)$. In this case, the $G$-structure on $U_{W_{r-1}, W_r} \cap U_{W_1, W_{2}}$ is given by $\Delta_{J_{1}}(C_*(W_r))^{-1} \sigma \Delta_{J_{r-1}}(C_*(W_r))$. This has positive determinant if and only if $\det(\sigma) > 0$. In otherwords, if and only if $\lfloor \frac{k}{2} \rfloor$ is odd.
\end{proof}

Below is an example from $\pi^{-1}(\cW_{2,6})$, where $W_{parallel} = W_{series}$.

\begin{eg} \label{2,6 example}
In this example, we explicitly write the diagrams in $H(W_{series})$ for $k=2$ and $n=6$. In what follows, the symbol $\bullet$ takes entries in $0$ or $1$, as in the definition of $C_*(W)$. Furthermore, variables are written with signs incorporated to force all Pl\"{u}cker coordinates to be positive. The boundaries between all cells defined by Wilson loop diagrams are documented in \cite{casestudy}. From there, we may directly see that the $\Sigma(W_i)$ indicated by the diagrams below each share a codimension one boundary with $\Sigma(W_{i+1})$.

\begin{tabular}{|c|c|}
  \hline
  $W_1 = \begin{tikzpicture}[rotate=60,baseline=(current bounding box.east)]
	\begin{scope}
	\drawWLD{6}{1}
	\drawprop{1}{0}{5}{0}
	\drawprop{2}{0}{4}{0}
		\end{scope}
	\end{tikzpicture}$
  &   $C_*(W_1) = \left(
          \begin{array}{ccccccc}
            c_{1,1} & c_{1,2} & 0 & 0 & -c_{1,5} & -c_{1,6} & \bullet \\
            0 & c_{2,2} & c_{2,3} & c_{2,4} & c_{2,5} & 0 & \bullet \\
          \end{array}
        \right)$ \\ \hline
    $W_2 = \begin{tikzpicture}[rotate=60,baseline=(current bounding box.east)]
	\begin{scope}
	\drawWLD{6}{1}
	\drawprop{6}{0}{4}{1}
	\drawprop{2}{0}{4}{-1}
		\end{scope}
	\end{tikzpicture}$
 &  $C_*(W_2) = \left(
          \begin{array}{ccccccc}
             c_{1,1} & 0 & 0 & -c_{1,4} & -c_{1,5} & -c_{1,6} & \bullet \\
            0 & c_{2,2} & c_{2,3} & c_{2,4} & c_{2,5} & 0 & \bullet \\
          \end{array}
        \right)$ \\ \hline
      $W_3 = \begin{tikzpicture}[rotate=60,baseline=(current bounding box.east)]
	\begin{scope}
	\drawWLD{6}{1}
	\drawprop{6}{0}{4}{0}
	\drawprop{1}{0}{3}{0}
		\end{scope}
	\end{tikzpicture}$ & $C_*(W_3) = \left(
          \begin{array}{ccccccc}
             c_{1,1} & 0 & 0 & -c_{1,4} & -c_{1,5} & -c_{1,6} & \bullet \\
            c_{2,1} & c_{2,2} & c_{2,3} & c_{2,4} & 0 & 0 & \bullet \\
          \end{array}
        \right)$  \\ \hline
    $W_4 = \begin{tikzpicture}[rotate=60,baseline=(current bounding box.east)]
	\begin{scope}
	\drawWLD{6}{1}
	\drawprop{6}{-1}{4}{0}
	\drawprop{6}{1}{2}{0}
		\end{scope}
	\end{tikzpicture}$ & $C_*(W_4) = \left(
          \begin{array}{ccccccc}
             c_{1,1} & 0 & 0 & -c_{1,4} & -c_{1,5} & -c_{1,6} & \bullet \\
             c_{2,1} & c_{2,2} & c_{2,3} & 0 & 0 & c_{2,6} & \bullet \\
          \end{array}
        \right)$ \\ \hline
  $W_5 = \begin{tikzpicture}[rotate=60,baseline=(current bounding box.east)]
	\begin{scope}
	\drawWLD{6}{1}
	\drawprop{5}{0}{3}{0}
	\drawprop{6}{0}{2}{0}
		\end{scope}
	\end{tikzpicture}$ & $C_*(W_5) = \left(
          \begin{array}{ccccccc}
             0 & 0 & -c_{1,3} & -c_{1,4} & -c_{1,5} & -c_{1,6} &  \bullet \\
             c_{2,1} & c_{2,2} & c_{2,3} & 0 & 0 & c_{2,6} & \bullet \\
          \end{array}
        \right)$  \\ \hline
   $W_6 = \begin{tikzpicture}[rotate=60,baseline=(current bounding box.east)]
	\begin{scope}
	\drawWLD{6}{1}
	\drawprop{5}{-1}{3}{0}
	\drawprop{5}{0}{1}{0}
		\end{scope}
	\end{tikzpicture}$ & $C_*(W_6) = \left(
          \begin{array}{ccccccc}
             0 & 0 & -c_{1,3} & -c_{1,4} & -c_{1,5} & -c_{1,6} &  \bullet \\
             c_{2,1} & c_{2,2} & 0 & 0 & c_{2,5} & c_{2,6} & \bullet \\
          \end{array}
        \right)$  \\ \hline
     $W_7 = \begin{tikzpicture}[rotate=60,baseline=(current bounding box.east)]
	\begin{scope}
	\drawWLD{6}{1}
	\drawprop{4}{0}{2}{0}
	\drawprop{5}{0}{1}{0}
		\end{scope}
	\end{tikzpicture}$ & $C_*(W_7) = \left(
          \begin{array}{ccccccc}
             0 & -c_{1,2} & -c_{1,3} & -c_{1,4} & -c_{1,5} & 0 & \bullet \\
             c_{2,1} & c_{2,2} & 0 & 0 & c_{2,5} & c_{2,6} & \bullet \\
          \end{array}
        \right)$  \\ \hline
\end{tabular}

For the the open sets $U_{W_1, W_2}$ and $U_{W_2, W_3}$, $J_1 = J_2 = \{1,2\}$. In this case, the transition matrix on the fibers is given by \bas t_{J_1, J_2} =
\left(\begin{array}{cc}
             c_{1,1} & 0  \\
            0 & c_{2,2} \\
          \end{array}
        \right)^{-1} \left(\begin{array}{cc}
             c_{1,1} & c_{1,2} \\
            0 & c_{2,2}  \\
          \end{array}\right)  \eas which has positive determinant.

For the open set $U_{W_6, W_7}$, $J_6 = \{1,3\}$. In this case, the transition matrix on the fibers is given by \bas t_{J_6, J_1} =
          \left(
          \begin{array}{cc}
             0 & -c_{1,2}  \\
             c_{2,1} & c_{2,2} \\
          \end{array}
        \right)^{-1}  \left(
          \begin{array}{cc}
             0 & 1  \\
             1 & 0 \\
          \end{array}
        \right) \left(\begin{array}{cc}
             0 & -c_{1,3} \\
             c_{2,1}  & 0  \\
          \end{array}
        \right) \;.\eas Notice that this has negative determinant.
\end{eg}

We conclude with the remark that Theorem \ref{thm:non-orientable}, inspired by conversations with Paul Heslop \cite{Heslopcommunication}, shows that $\pi^{-1}(\cW)$ is not orientable in approximately $3$ out of $4$ pairs $(k,n)$ with $n\geq k+4$. The authors are confident that with a better understanding of the boundary structures between positroid cells defined by Wilson Loop diagrams, one may show that all such bundles are non-orientable.

Finally, while volume forms are not well defined on non-orientable manifolds, the authors do not believe that the nonorientability of the space poses a threat to the program of geometrically understanding the Wilson loop amplitudes geometrically. Rather, we hope that this makes the problem more subtle and interesting. It is possible that the integrals associated to the diagrams correspond to some sort of characteristic class of the manifold, and that the volumes associated to the Amplituhedron may be seen as a special case of this more general result.


\end{document}